\documentclass[lettersize,journal]{IEEEtran}
\usepackage{amsmath,amsfonts,amssymb}
\usepackage[ruled,vlined,procnumbered,linesnumbered]{algorithm2e}
\usepackage{array}
\usepackage[dvipsnames, svgnames]{xcolor}
\usepackage{yhmath}
\usepackage{soul}
\usepackage{amsthm}
\newtheorem{theorem}{Theorem}
\newtheorem{lemma}{Lemma}
\usepackage[scr=dutchcal,scrscaled=1.05]{mathalfa}
\usepackage{textcomp}
\usepackage{stfloats}
\usepackage{url}
\usepackage{verbatim}
\usepackage{graphicx}
\usepackage{subfigure}
\usepackage{cite}
\usepackage[colorlinks, linkcolor=blue, citecolor=blue]{hyperref}

\begin{document}

\title{Optimal Beamforming for Uplink Covert Communication in MIMO GEO\\ Satellite-Terrestrial Systems}

\author{Zewei Guo,~\IEEEmembership{Graduate Student Member,~IEEE,} Ranran Sun, Yulong Shen,~\IEEEmembership{Member,~IEEE,} \\ and Xiaohong Jiang,~\IEEEmembership{Senior Member,~IEEE}
}

\markboth{}%
{Shell \MakeLowercase{\textit{et al.}}: A Sample Article Using IEEEtran.cls for IEEE Journals}


\maketitle

\begin{abstract}
This paper investigates the uplink covert communication in a multiple-input multiple-output (MIMO) satellite-terrestrial system consisting of an Earth station transmitter Alice, a geosynchronous Earth orbit (GEO) satellite receiver Bob, and multiple GEO satellite wardens around Bob, where each node in the system is equipped with an array of directional antennas. Based on beamforming and the default antenna orientation setting, we first propose a scheme for covert Alice-Bob uplink transmission. Under the perfect channel estimation scenario, we provide theoretical modeling for the system performance in terms of detection error probability (DEP), transmission outage probability (TOP) and covert rate (CR), and then explore the optimal beamforming (OB) design as well as the joint optimal beamforming and antenna orientation (JO-BA) design for CR maximization. We then extend our study to the imperfect channel estimation scenario, and conduct related performance modeling and OB/JO-BA designs for CR maximization. We also apply the techniques of semidefinite relaxation, alternating optimization, Rodrigues' rotation formula and 1-D search algorithm to develop efficient algorithms to solve the above optimization problems. Finally, extensive numerical results are presented to verify our theoretical results and to illustrate the efficiency of beamforming and antenna orientation design for supporting the uplink covert communication in MIMO GEO satellite-terrestrial systems.
\end{abstract}
\begin{IEEEkeywords}
Covert transmission, beamforming, MIMO communication, multiple wardens, satellite communication systems, directional antennas, Rician channel.
\end{IEEEkeywords}

\section{Introduction}
\IEEEPARstart{S}{atellite-terrestrial} systems rely on satellite constellations to provide seamless and ubiquitous wireless connectivity worldwide, which is particularly urgent for terrestrial users in remote and underserved areas without terrestrial network coverage \cite{Kodheli2021}. Such systems can complement and extend the existing terrestrial wireless communication networks to achieve enhanced throughput, improved energy efficiency and broader coverage, so they are recognized as an indispensable component for future Sixth-Generation (6G) networks to support various critical applications like disaster recovery, maritime communications, and global navigation \cite{Zhu2021}.

Due to the openness of space and the broadcast nature of wireless communications, satellite-terrestrial systems face serious security threats caused by various attacks from potential adversaries. Covert communication (or low probability of detection communication), which hides communication signals in background noise to conceal the presence of wireless transmission from malicious wardens, is a highly promising solution to provide a stronger form of security guarantee for satellite-terrestrial systems. Specially, covert communication is essential for satellite-terrestrial systems to support some security-critical applications, such as sensitive telemetry, coastline surveillance, and deep-sea exploration~\cite{An2024}. Therefore, how to achieve efficient covert communication in satellite-terrestrial systems now becomes an increasingly important research topic in both academia and industry\cite{chen2023,Guo2024}.

Some works have been conducted on the downlink covert communication in satellite-terrestrial systems~\cite{Yu2024,Mu2025,Jia2025,Feng2024,Song2023}. 
The authors in \cite{Yu2024} develop a random variance technique for covert signals to enable covert communication in low Earth orbit (LEO) satellite-terrestrial systems, and further study the optimal random variance design for covert signals to improve the covert rate. The works in \cite{Mu2025} and \cite{Jia2025} explore the jamming-based covert communication schemes for satellite-terrestrial systems. In particular, \cite{Mu2025} designs both the uninformed jamming scheme and cognitive jamming scheme for covert communication in space-air-ground integrated networks, while \cite{Jia2025} applies the techniques of rate-splitting multiple access and jamming to develop a novel covert communication scheme for a LEO satellite-terrestrial system. The work in \cite{Feng2024} proposes a power control scheme for covert communication in a large-scale LEO satellite-terrestrial system with multiple satellites and multiple terrestrial base stations. 
The work in \cite{Song2023} explores the reconfigurable intelligent surface (RIS)-assisted covert communication in a geosynchronous Earth orbit (GEO) satellite-terrestrial system, and further identifies the optimal design of RIS phase shifts and covert transmit power for covert rate maximization in the system.

It is notable that the uplink transmission in satellite-terrestrial systems is also crucial for various services like emergency communication, disaster response, and Internet connectivity~\cite{Bakhsh2024}. 
Recently, some initial works are devoted to the study of uplink covert communication in satellite-terrestrial systems \cite{Wang2022,Yu2025}.
The authors in \cite{Wang2022} develop an uplink covert communication scheme for space-air-ground integrated networks, and also explore the joint optimal design of transmit power and covert signaling factor for covertness performance enhancement there.
The work in \cite{Yu2025} deals with the uplink covert communication in a GEO satellite-UAV system, and examines the joint 3D beamforming and UAV trajectory design for covert rate maximization in such a system.

Although the aforementioned works provide valuable insights into uplink covert communication in satellite-terrestrial systems, several critical issues remain largely unexplored. First, these works mainly focus on the single-input single-output (SISO) or multiple-input single-output (MISO) satellite–terrestrial systems. Recently, multiple-input multiple-output (MIMO) techniques have been widely employed in satellite–terrestrial systems to provide reliable and high-throughput uplink communication \cite{Heo2023}, but how to take advantage of MIMO for more efficient uplink covert communication in satellite–terrestrial systems has not been investigated yet.
Second, the above studies consider only the single-warden scenario. In a practical satellite–terrestrial system, multiple wardens may exist to perform more aggressive detection attacks. However, how to achieve effective uplink covert communication under this practical multi-warden scenario remains an unexplored issue.
Third, the available uplink covert communication designs are based on the perfect estimation of channel state information (CSI). Due to the intrinsic channel impairments (e.g., propagation delays and frequency offsets), the channel estimation in satellite-terrestrial systems may be imperfect~\cite{Li2022}, while the uplink covert communication design under imperfect channel estimation remains to be explored.

As a first attempt to address the above issues, this paper considers a general MIMO GEO satellite-terrestrial system consisting of an Earth station transmitter, a GEO satellite receiver, and multiple GEO satellite wardens around the receiver, where each node in the system is equipped with an array of directional antennas. For this system, we investigate the optimal beamforming design to support uplink covert communication under both perfect and imperfect channel estimations. The contributions of this paper are as follows:
\begin{itemize}
\item 
{We focus on the uplink covert communication in a MIMO satellite-terrestrial system consisting of an Earth station transmitter Alice, a satellite receiver Bob, and multiple satellite wardens around Bob, where each node in the system is equipped with an array of directional antennas. By exploring the beamforming design and default antenna orientation setting, we first propose a scheme for covert Alice-Bob uplink transmission in the considered system. }
\item 
{Under the perfect channel estimation scenario, we provide theoretical modeling for the detection error probability (DEP), transmission outage probability (TOP), and covert rate (CR), and then investigate the optimal beamforming (OB) design as well as the joint optimal beamforming and antenna orientation (JO-BA) design for CR maximization.
Based on the techniques of semidefinite relaxation, alternating optimization, Rodrigues’ rotation formula and 1-D search algorithm, efficient algorithms are developed to solve these optimization problems.
}
\item
{Considering channel estimation errors, we further explore the related performance modeling and the CR maximization problems under the imperfect channel estimation scenario, and also develop the related efficient algorithms to solve these maximization problems.
}
\item
{Finally, we present extensive numerical results to verify our theoretical results and to illustrate the efficiency of beamforming and antenna orientation design for supporting the uplink covert communication in MIMO satellite-terrestrial systems.}
\end{itemize}

{The remainder of this paper is organized as follows. Section~\ref{sec:system_model} gives the system model and preliminaries. 
Section~\ref{sec:covert_scheme_perfect} introduces the OB and JO-BA designs under perfect channel estimation, while Section~\ref{sec:covert_scheme_imp} presents these two designs under imperfect channel estimation.
Section~\ref{sec:Simu} presents simulation and numerical results. The conclusions are drawn in Section~\ref{sec:Conclusion}}.

\emph{Notation}: Lower-case, lower-case bold-face, and upper-case bold-face letters represent scalars, vectors, and matrices (e.g., $a$, $\mathbf{a}$ and $\mathbf{A}$), respectively. $\otimes$, ${\rm log}(\cdot)$, ${\rm log}_2(\cdot)$, and ${\rm log}_{10}(\cdot)$ are the Kronecker product operator, the natural logarithm, the binary logarithm, and the common logarithm, respectively. ${\rm Pr}\{\cdot \}$ and $\mathbb{E}[\cdot]$ denote the probability and the expectation operators. Let $\mathrm{Tr}({\mathbf A})$, ${\mathbf A}^{-1}$, $||{\mathbf A}||$, ${\mathbf A}^{T}$, and ${\mathbf A}^\dagger$ denote the trace, inverse, Frobenius norm, transpose, and conjugate transpose of matrix ${\mathbf A}$, respectively.
$\mathbf{A}\succeq \mathbf{0}$ denotes that matrix $\mathbf{A}$ is a semi-definite matrix. 
Denote $\mathbb{C}^{M \times N}$ as the $M\times N$ dimensional complex-valued matrix space, $\mathbf I_{M}$ as the $M\times M$ identity matrix, and $|\mathscr{c}|$ denotes the modulus of the complex number~$\mathscr{c}$. 

\begin{figure}[tbp]
\centering
\includegraphics[width=0.9\linewidth]{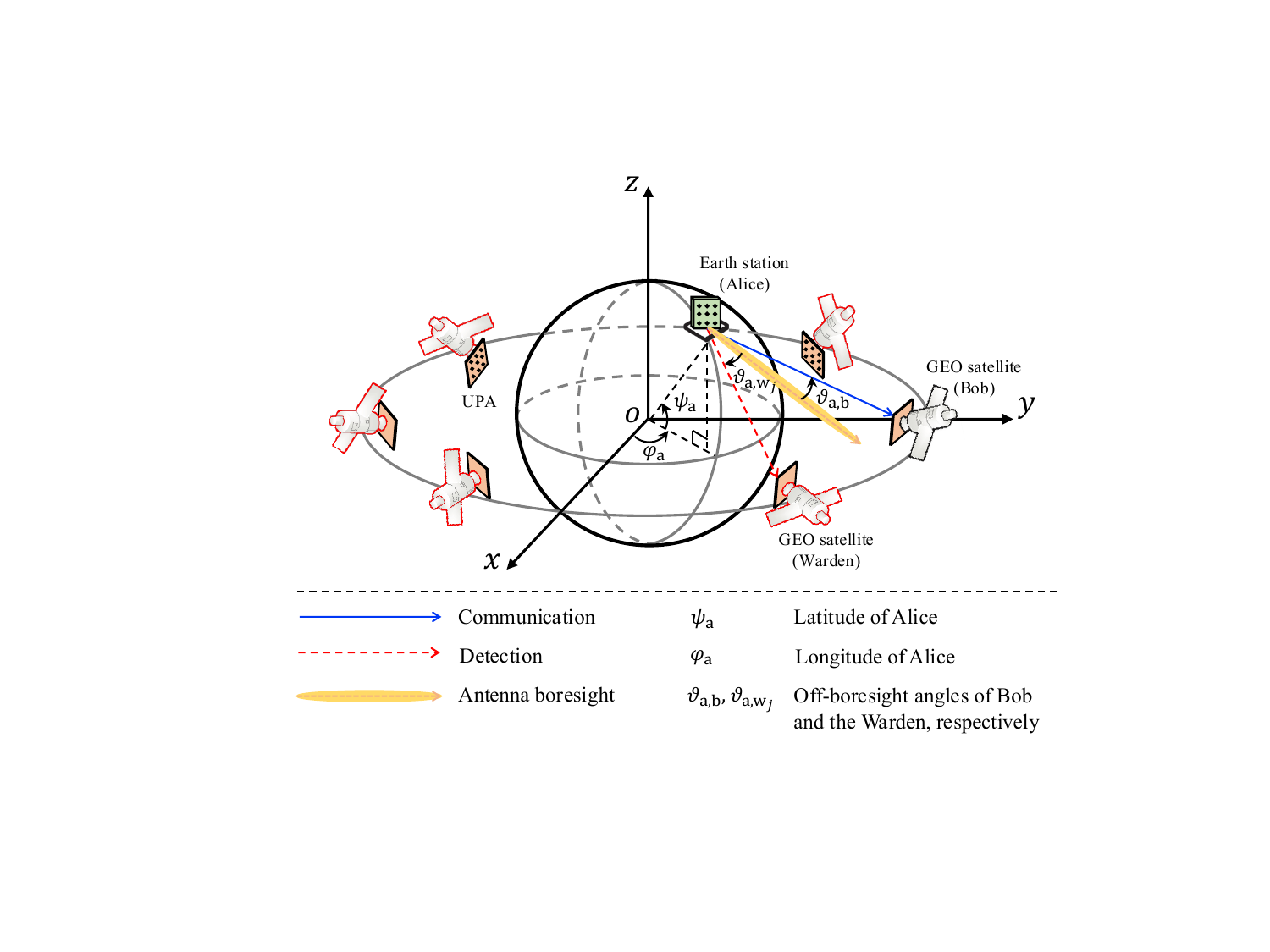}
\vspace{-12pt}
\caption{System Model.}
\label{system_model}
\vspace{-15pt}
\end{figure}

\vspace{-5pt}
\section{System Model and Preliminaries}
\label{sec:system_model}
\vspace{-3pt}
\subsection{System Model}
\vspace{-2pt}
As illustrated in Fig. \ref{system_model}, we consider a MIMO satellite-terrestrial covert communication system consisting of an Earth station transmitter Alice (``$\mathrm{a}$"), a GEO satellite receiver Bob~(``$\mathrm{b}$"), and $N_\mathrm{w}$ GEO satellite wardens around Bob. 
Alice is equipped with a uniform planar array (UPA) of $M_{\mathrm{a}}=M_{\mathrm{a},x}\times M_{\mathrm{a},z}$ directional antenna elements. Both Bob and each warden are equipped with a UPA of $M_{\mathrm{sat}} = M_{\mathrm{sat},x}\times M_{\mathrm{sat},z}$ directional antenna elements. 
In the system, Alice tries to perform covert uplink transmission to Bob by utilizing the beamforming technique, while $n_\mathrm{w}$ ($1\leq n_\mathrm{w} \leq N_\mathrm{w}$) wardens simultaneously receive the signal from Alice to determine whether the Alice-Bob transmission exists or not.
Both Bob and the $n_\mathrm{w}$ wardens are assumed to be located within Alice’s visible region.
To model the location of all nodes in the considered system, we utilize the three-dimensional Cartesian coordinate system by setting the center of the Earth as the origin. We use $\mathbf{q}_\mathrm{a}$, $\mathbf{q}_\mathrm{b}$, and $\mathbf{q}_{\mathrm{w}_j}$ to denote the locations of Alice, Bob, and the $j^{\mathrm{th}}$ warden ($1\leq j \leq n_\mathrm{w}$), respectively.
Following \cite{Jung2024}, the locations of Alice and the satellite $s\in\{\mathrm{b},
\mathrm{w}_j\}$ are given~by,

\vspace{-3mm}
\begin{footnotesize}
\begin{align}
        \mathbf{q}_\mathrm{a} &= d_{\mathrm{a}}\left[
 \cos \psi_\mathrm{a} \cos \varphi_\mathrm{a},\,
 \cos \psi_\mathrm{a} \sin \varphi_\mathrm{a},\,
 \sin \psi_\mathrm{a} \right]^T,\\
    \mathbf{q}_s &= d_{\mathrm{sat}}\left[
  \cos \varphi_s, \,
  \sin \varphi_s, \,
0  \right]^T,
\end{align}    
\end{footnotesize}where $d_{\mathrm{a}}$ is the radius of the Earth, $d_{\mathrm{sat}} \triangleq d_{\mathrm{a}} + h$ denotes the distance from the Earth’s center to the satellites, $h$ is the altitude of the GEO satellites, $\psi_\mathrm{a}$ denotes the latitude of Alice, $\varphi_\mathrm{a}$ and $\varphi_s$ represent the longitude of Alice and the satellite $s$, respectively. Consider a worst-case scenario for covert transmission, all $n_\mathrm{w}$ wardens are assumed to have access to the exact locations of Alice and Bob through passive localization technologies, such as optical rangefinders~\cite{Medvedev2001} and satellite pass predictors~\cite{Vallado2008}.

\vspace{-7pt}
\subsection{Channel Model}
\vspace{-2pt}
The concerned system involves multiple satellite–terrestrial uplinks.
To model the uplink between the Earth station Alice and the satellite $s$, the antenna patterns of Alice and the satellite $s$, path loss, and channel fading are taken into account, which are modeled as follows.

1) \textbf{Earth station antenna pattern} refers to the antenna gain of Alice in different orientations. We use ${G_{\mathrm{a},s}^{\text{dBi}}}$ to denote the antenna gain of Alice in the direction of the satellite~$s$, measured in decibels relative to isotropic (dBi),
which is given by~[\citenum{ITU465}, P. 2] 
\begin{equation}
\footnotesize
	{G}_{\mathrm{a},s}^{\text{dBi}}=
	\begin{cases}
		G_{\mathrm{a},\mathrm{max}}^{\text{dBi}}, & 0 \leq \vartheta_{\mathrm{a},s}<\vartheta_{0}, \\
G_{\mathrm{a},\mathrm{max}}^{\text{dBi}} - 25\log_{10} \left( {\vartheta_{\mathrm{a},s}} \right),  & {\vartheta_{0}}\leq \vartheta_{\mathrm{a},s} \leq 48^{\circ}, \\
-10, &  48^{\circ}<\vartheta_{\mathrm{a},s}<180^{\circ},
	\end{cases}
\label{antenna_gain_alice}
\end{equation}
where $G_{\mathrm{a},\mathrm{max}}^{\text{dBi}}$ is the maximum antenna gain of the Earth station Alice, $\vartheta_{0}$ is the minimum off-axis angle, {$\vartheta_{\mathrm{a},s}$ is the off-boresight angle of the satellite $s$ with respect to the boresight of Alice. Here, $\vartheta_{\mathrm{a},s}$ is expressed as}
\begin{equation}
\footnotesize
    \vartheta_{\mathrm{a},s} = \mathrm{arccos}\left(\frac{\mathbf{o}_{\mathrm{a}} \cdot \mathbf{u}_{\mathrm{a},s}}{\left \|\mathbf{o}_{\mathrm{a}}\right \| \cdot \left \|\mathbf{u}_{\mathrm{a},s}\right \|}\right),
\end{equation}
where $\mathbf{o}_{\mathrm{a}}$ is the antenna boresight vector of Alice and $\mathbf{u}_{\mathrm{a},s}\triangleq\mathbf{q}_s - \mathbf{q}_\mathrm{a}$ denotes the vector from Alice to the satellite~$s$.

2) \textbf{Satellite antenna pattern} refers to the satellite antenna gain in different orientations. We use $G_{s,\mathrm{a}}^{\text{dBi}}$ to denote the receiver antenna gain of the satellite $s$, measured in dBi, which can be expressed as~[\citenum{ITU672}, P. 2]
\begin{equation}
\footnotesize
    \label{antenna_gain_satellite}
    {G_{s,\mathrm{a}}^{\text{dBi}}}=
    \begin{cases}
        G_{\mathrm{sat},\mathrm{max}}^{\text{dBi}}-3\left(\dfrac{\varPhi_{s,\mathrm{a}}}{\varPhi_{3{\text {dB}}}} \right)^2, & \hspace{-5pt} 0 \leq \varPhi_{s,\mathrm{a}}<\alpha \varPhi_{3{\text {dB}}}, \\
        G_{\mathrm{sat},\mathrm{max}}^{\text{dBi}}+L_S,  & \hspace{-5pt} \alpha\varPhi_{3{\text {dB}}}\leq \varPhi_{s,\mathrm{a}} \leq \beta\varPhi_{3{\text {dB}}}, \\
        G_{\mathrm{sat},\mathrm{max}}^{\text{dBi}}+L_S-25 \log_{10} \left( \dfrac{\varPhi_{s,\mathrm{a}}}{\varPhi_{3{\text {dB}}}} \right),  & \hspace{-5pt} \beta\varPhi_{3{\text {dB}}}\leq \varPhi_{s,\mathrm{a}} \leq 90^{\circ}, \\
        L_F, & \hspace{-5pt} 90^{\circ}<\varPhi_{s,\mathrm{a}}<180^{\circ},
    \end{cases}
\end{equation}
where $G_{\mathrm{sat},\mathrm{max}}^{\text{dBi}}$ is the maximum antenna gain of the satellite's UPA, $\varPhi_{3{\text {dB}}}$ is one-half the $3{\text {dB}}$ beamwidth of the satellite antenna, $\alpha$ and $\beta$ are the antenna parameters,  $L_S$ and $L_F$ are the near-in side-lobe level and the far-out side-lobe level, respectively, {and $\varPhi_{s,\mathrm{a}}$ is the off-boresight angle of Alice with respect to the boresight of the satellite $s$.
Here, $\varPhi_{s,\mathrm{a}}$ is given~by}

\begin{equation}
\footnotesize
    \varPhi_{s,\mathrm{a}} = \mathrm{arccos}\left(\frac{\mathbf{o}_{s} \cdot \mathbf{u}_{s,\mathrm{a}}}{\left \|\mathbf{o}_{s}\right \| \cdot \left \|\mathbf{u}_{s,\mathrm{a}}\right \|}\right),
\end{equation} 
where $\mathbf{o}_{s}$ is the antenna boresight vector of the satellite $s$ and $\mathbf{u}_{s,\mathrm{a}}\triangleq\mathbf{q}_\mathrm{a} - \mathbf{q}_s$ is the vector from the satellite $s$ to Alice.

3) \textbf{Free space loss} denotes the attenuation of signal propagation along the path from Alice to the satellite $s$, which can be given by $F_{\mathrm{a},s} = \left({\lambda}/{4 \pi d_{\mathrm{a},s}}\right)^2$,
where $d_{\mathrm{a},s}$ is the distance from Alice to the satellite $s$.

\begin{figure}[tb]
\centering
\includegraphics[width=0.85\linewidth]{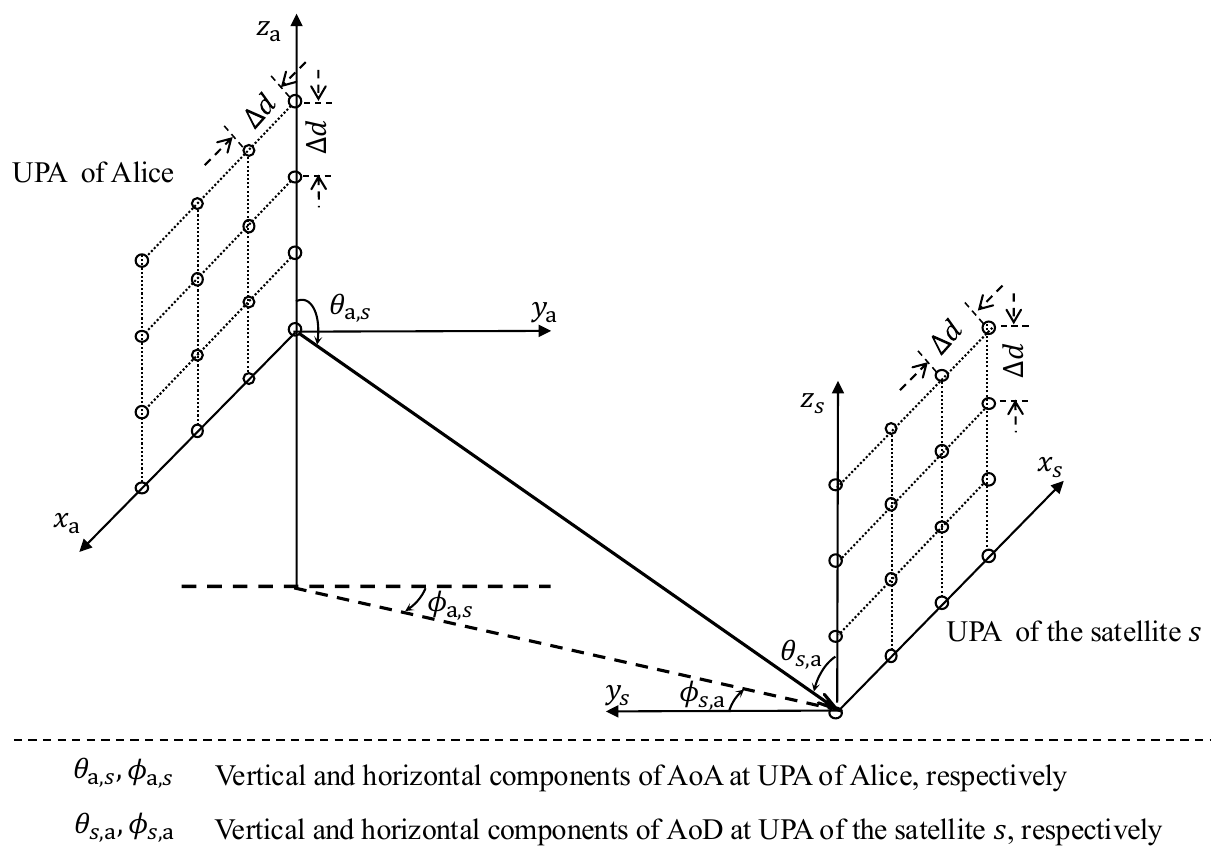}
\vspace{-10pt}
\caption{Geometrical relation between UPA of Alice and that of any satellite.}
\vspace{-17pt}
\label{UPA_model}
\end{figure}

4) \textbf{Channel fading} is caused by multipath effects. 
Same as \cite{Yu2025}, we consider all satellite-terrestrial uplinks in the considered system as quasi-static Rician fading channels, which remain constant in a time slot while changing independently from one time slot to another. Let ${\bf H}_{\mathrm{a},s} \in \mathbb{C}^{M_{\mathrm{sat}} \times M_{\mathrm{a}}}$ denote the channel fading matrix from Alice to the satellite $s$, then ${\bf H}_{\mathrm{a},s}$ is determined by
\begin{equation}
\footnotesize
\label{eq:Rician_channel}
	\mathbf{H}_{\mathrm{a},s}=\sqrt{\frac{K}{K+1}} \overline{\mathbf{H}}_{\mathrm{a},s}+\sqrt{\frac{1}{K+1}} \widetilde{\mathbf{H}}_{\mathrm{a},s}.
\end{equation} 
Here, $K$ denotes the Rician factor, which represents the ratio of the signal power in the line-of-sight (LoS) component to that in the multipath scattered component. 
$\overline{\mathbf{H}}_{\mathrm{a},s}$ and $\widetilde{\mathbf{H}}_{\mathrm{a},s}$ denote the channel fading matrices of the LoS component and the multipath scattered component, respectively. 
Due to the randomness of multipath scattering, the entries of $\widetilde{\mathbf{H}}_{\mathrm{a},s}$ are independent and identically distributed (i.i.d.) complex Gaussian random variables with zero mean and unit variance.
As illustrated in Fig.~\ref{UPA_model}, both Alice and the satellite $s$ are equipped with UPA antennas. The channel fading matrix of the LoS component $\overline{\mathbf{H}}_{\mathrm{a},s} \in \mathbb{C}^{M_{\mathrm{sat}}\times M_{\mathrm{a}}}$ is determined by
\begin{equation}
\footnotesize
    \overline{\mathbf{H}}_{\mathrm{a},s} = \mathbf{g}_{s,\mathrm{a}}\mathbf{d}_{\mathrm{a},s}^\dagger,
\end{equation}
where $\mathbf{g}_{s,\mathrm{a}} \in \mathbb{C}^{M_{\mathrm{sat}}\times1}$ is the steering vector of the UPA of the satellite $s$, formulated as
\begin{equation}
    \footnotesize
    \mathbf{g}_{s,\mathrm{a}} =\mathbf{a}_{M_{\mathrm{sat},x}}\left(\sin\theta_{s,\mathrm{a}}\right) \otimes \mathbf{a}_{M_{\mathrm{sat},z}}\left(\cos\theta_{s,\mathrm{a}}\sin\phi_{s,\mathrm{a}}\right).
\end{equation}
Here, $\theta_{s,\mathrm{a}}$ and $\phi_{s,\mathrm{a}}$ denote the vertical and horizontal components of the angle of arrival (AoA) at the UPA of the satellite~$s$, respectively. $\mathbf{a}_{M_{\mathrm{sat},x}}(\sin\theta_{s,\mathrm{a}})$ and $\mathbf{a}_{M_{\mathrm{sat},z}}(\cos\theta_{s,\mathrm{a}}\sin\phi_{s,\mathrm{a}})$ denote the corresponding horizontal and vertical components of the steering vector of the UPA of the satellite~$s$, respectively. $\mathbf{a}_{n_v}\left(x\right)$ ($n_v\in \{M_{\mathrm{sat},x},M_{\mathrm{sat},z},M_{\mathrm{a},x},M_{\mathrm{a},z}\}$) denotes the spatial signature vector at the UPA, which is given by
\begin{equation}
    \footnotesize
    \mathbf{a}_{n_v}\left(x\right)=\Big[1, e^{-{{\jmath}} 2 \pi \frac{\Delta d}{\lambda} x}, \ldots, 
        e^{-{{\jmath}} 2 \pi(n_v-1) \frac{\Delta d}{\lambda} x} \Big]^T,
\end{equation}
where ${\jmath}=\sqrt{-1}$, $\lambda$ is the carrier wavelength, and $\Delta d$ is the distance between adjacent antenna elements.
Similarly, $\mathbf{d}_{\mathrm{a},s}\in \mathbb{C}^{M_{\mathrm{a}}\times1} $ refers to the steering vector of the UPA of Alice, which can be given by
\begin{equation}
    \footnotesize
    \mathbf{d}_{\mathrm{a},s} =\mathbf{a}_{M_{\mathrm{a},x}}\left(\sin\theta_{\mathrm{a},s}\right) \otimes \mathbf{a}_{M_{\mathrm{a},z}}\left(\sin\theta_{\mathrm{a},s}\sin\phi_{\mathrm{a},s}\right),
\end{equation}
where $\theta_{\mathrm{a},s}$ and $\phi_{\mathrm{a},s}$ are the vertical and horizontal components of the angle of departure (AoD) of Alice’s UPA, respectively. $\mathbf{a}_{M_{\mathrm{a},x}}\left(\sin\theta_{\mathrm{a},s}\right)$ and $\mathbf{a}_{M_{\mathrm{a},z}}\left(\sin\theta_{\mathrm{a},s}\sin\phi_{\mathrm{a},s}\right)$ are the corresponding horizontal component and vertical component of the steering vector of Alice's UPA, respectively.

If we use ${\bf C}_{\mathrm{a},s}\in \mathbb{C}^{M_{\mathrm{sat}}\times M_{\mathrm{a}}}$ to denote the channel coefficient matrix of the satellite-terrestrial uplink between Alice and the satellite $s$, then ${\bf C}_{\mathrm{a},s}$ is determined by
\begin{equation}
	{\bf C}_{\mathrm{a},s} = \sqrt{F_{\mathrm{a},s}G_{\mathrm{a},s}G_{s,\mathrm{a}}}{\bf H}_{\mathrm{a},s},
	\label{s-t channel model}
\footnotesize
\end{equation}
where ${G}_{\mathrm{a},s}=10^{{G_{\mathrm{a},s}^{\text{dBi}}}/10}$ and ${G}_{s,\mathrm{a}}=10^{{G_{s,\mathrm{a}}^{\text{dBi}}}/10}$.

Following \cite{Wang2021,Wang2024}, we assume that Bob periodically broadcasts pilot signals to Alice before transmission, so that Alice can estimate the channel coefficient matrices of the Alice-Bob link (i.e., ${\bf C}_{\mathrm{a},\mathrm{b}}$). Since Alice and the wardens remain silent during the entire channel estimation process, they can only acquire the statistical CSI of the link from Alice to the $j^{\mathrm{th}}$ warden (i.e., ${\bf C}_{\mathrm{a},\mathrm{w}_j}$).

\vspace{-8pt}
\subsection{Noise Uncertainty at Bob and Wardens}
\vspace{-3pt}

In practice, many factors such as temperature variation, calibration errors, and environmental changes make the accurate estimation of background noise difficult~\cite{He2017}.
Note that the noise uncertainty at Bob {affects the reliability of the Alice-Bob covert transmission}, while that at the wardens {impairs their ability to detect} the covert transmission.
To this end, we consider the background noise uncertainty at Bob and the wardens, which can be modeled by the bounded uncertainty model \cite{He2017}.
Let $\sigma_s^2$ denote the received noise power at the satellite~$s$, which {lies within a finite interval around} the nominal value $\overline{\sigma}_s^2$ and {follows a log-uniform distribution} over $[\overline{\sigma}_s^2/\rho,\,\rho\overline{\sigma}_s^2]$.
Here, $\rho$ quantifies the size of the noise uncertainty and it satisfies $\rho\geq1$.
Thus, the probability density function (PDF) of $\sigma_s^2$ {is given by}
\begin{equation}
\label{eq:noise_PDF}
\footnotesize
f_{\sigma_{s}^2}(x)= \begin{cases}\frac{1}{2 \log (\rho) x}, &  \hat{\sigma}_{s}^{\mathrm{lb}} < x <  \hat{\sigma}_{s}^{\mathrm{ub}}, \\ 0, & \text {otherwise},\end{cases}
\end{equation}
where $\hat{\sigma}_{s}^{\mathrm{lb}} \triangleq \overline{\sigma}_s^2/\rho$ and $\hat{\sigma}_{s}^{\mathrm{ub}} \triangleq \rho\overline{\sigma}_s^2$.

\vspace{-8pt}
\subsection{Detection Model}
\vspace{-3pt}
In the considered system, the $n_\mathrm{w}$ wardens aim to determine whether Alice transmits to Bob or not based on their received signals.
Each warden performs a binary hypothesis test {consisting of} a null hypothesis $\mathcal{H}_0$ and an alternative hypothesis $\mathcal{H}_1$.
The former $\mathcal{H}_0$ indicates that Alice does not covertly transmit to Bob, while the latter one $\mathcal{H}_1$ indicates that Alice does the covert transmission. 
{Assume that the LoS component-based equal gain combining (LoS-EGC) technique is adopted at each satellite to combine the received signal, which has low computational complexity and does not require instantaneous CSI of transmission links~\cite{Yue2015,Jin2018}. 
Let $\mathbf{v}_{\mathrm{w}_j}=\mathbf{g}_{\mathrm{w}_j,\mathrm{a}}/\sqrt{M_{\mathrm{sat}}}$ denote the LoS-EGC vector at the $j^{\mathrm{th}}$ warden.}
Then, the received signals ${y}_{\mathrm{w}_j}[k]$ at the $j^{\mathrm{th}}$ warden can be expressed~as
\begin{equation}
\footnotesize
{y}_{\mathrm{w}_{j}}[k]= 
\begin{cases}
\mathbf{v}_{\mathrm{w}_j}^{\dagger}\mathbf{n}_{\mathrm{w}_{j}}[k], & \mathcal{H}_0, \\[2mm]  
    \sqrt{F_{\mathrm{a},\mathrm{w}_{j}} G_{\mathrm{a},\mathrm{w}_{j}}G_{\mathrm{w}_{j},\mathrm{a}}P_\mathrm{a}}\mathbf{v}_{\mathrm{w}_j}^{\dagger}\mathbf{H}_{\mathrm{a},\mathrm{w}_{j}}\mathbf{w}_\mathrm{a} f_{\mathrm{a}}[k] \\ \quad+ \mathbf{v}_{\mathrm{w}_j}^{\dagger}\mathbf{n}_{\mathrm{w}_{j}}[k], & \mathcal{H}_1,
    \label{willie_recieve}
\end{cases}
\end{equation}
where $k=1,2,\dots,N$ {denotes the time-slot index}. $\mathbf{n}_{\mathrm{w}_j}[k]\in \mathbb{C}^{M_{\mathrm{sat}} \times 1}$ {denotes} the additive white Gaussian noise {vector} at {the} $j^{\mathrm{th}}$ warden, {where each element of $\mathbf{n}_{\mathrm{w}_j}[k]$ is i.i.d. circularly symmetric complex Gaussian random variables with zero mean and variance} $\sigma_{\mathrm{w}_j}^{2}$, i.e., $\mathbf{n}_{\mathrm{w}_j}[k] \sim \mathcal{CN}(\mathbf{0}, \sigma_{\mathrm{w}_j}^2 \mathbf{I}_{M_{\mathrm{sat}}})$.
$P_\mathrm{a}$ is the transmit power of Alice. 
$\mathbf{w}_\mathrm{a} \in \mathbb{C}^{M_{\mathrm{a}} \times 1}$ is the beamforming vector of Alice and satisfies $\|\mathbf{w}_\mathrm{a}\|=1$.
$f_{\mathrm{a}}[k]$ is the covert signal transmitted by Alice at the $k^{\mathrm{th}}$ time slot and satisfies $\mathbb{E}[|f_{\mathrm{a}}[k]|^2]=1$.

According to~(\ref{willie_recieve}), the average received power $T_{\mathrm{w}_j}$ at the $j^{\mathrm{th}}$ warden is denoted by~$T_{\mathrm{w}_j} \triangleq \frac{1}{N} \sum_{k=1}^{N}\left|{y}_{\mathrm{w}_j}[k]\right|^2$. Considering an infinite number of signal observations at each warden, by using the strong law of large numbers, $T_{\mathrm{w}_j}$ is given by~\cite{Xia2024,Guo2025,Guo2025a}
\begin{equation}
\footnotesize
    T_{\mathrm{w}_j} \to 
        \begin{cases} \sigma^2_{\mathrm{w}_{j}}, & \mathcal{H}_0, \\ 
        S_{\mathrm{w}_{j}}+\sigma^2_{\mathrm{w}_{j}}, & \mathcal{H}_1,
        \end{cases}
    \label{Tw}
\end{equation}
{where $S_{\mathrm{w}_{j}}\triangleq{F_{\mathrm{a},\mathrm{w}_{j}} G_{\mathrm{a},\mathrm{w}_{j}}G_{\mathrm{w}_{j},\mathrm{a}} P_\mathrm{a}}|\mathbf{v}_{\mathrm{w}_j}^{\dagger}\mathbf{H}_{\mathrm{a},\mathrm{w}_{j}}\mathbf{w}_\mathrm{a}|^2$ denotes the received signal power at the $j^{\mathrm th}$ warden.}

To detect the covert communication, each warden performs a threshold test based on $T_{\mathrm{w}_j}$~[\citenum{Xia2024}, Eq. (6)], which is given by 
$T_{\mathrm{w}_j}  \underset{\mathcal{D}_0}{\stackrel{\mathcal{D}_1}{\gtrless}} \tau_{j}$. Here, $\tau_{j}$ denotes the detection threshold of the $j^{\mathrm{th}}$ warden, $\mathcal{D}_0$ and $\mathcal{D}_1$ are the decisions that support $\mathcal{H}_0$ and $\mathcal{H}_1$, respectively.
If $T_{\mathrm{w}_j}>\tau_j$, the warden makes a decision $\mathcal{D}_1$ that Alice performs the covert transmission, otherwise, the warden decides $\mathcal{D}_0$ that Alice does not perform the covert transmission. 
However, two types of errors may occur during the detection, which are defined as the false alarm and the missed detection. 
The former one represents the event that the warden makes a decision $\mathcal{D}_1$ in favor of $\mathcal{H}_1$ when $\mathcal{H}_0$ is true, while the latter one represents the event that the warden makes a decision $\mathcal{D}_0$ in favor of $\mathcal{H}_0$ when $\mathcal{H}_1$ is true. The probabilities of false alarm and missed detection at the $j^{\mathrm{th}}$ warden are defined as $\mathbb{P}_{\mathrm{F}_j} \triangleq \Pr\{\mathcal{D}_1 | \mathcal{H}_0\}$ and $\mathbb{P}_{\mathrm{M}_j} \triangleq \Pr\{\mathcal{D}_0 | \mathcal{H}_1\}$, respectively.

\vspace{-9pt}
\subsection{Covert Communication Scheme}
\vspace{-3pt}
This subsection presents the proposed uplink covert communication scheme. First, the CSI of the Alice-Bob channel and {the LoS CSI of the channels from Alice to each of the $n_\mathrm{w}$ wardens} are estimated.
Based on the estimated CSI, the beamforming vector for Alice-Bob covert transmission is designed. Finally, the covert uplink transmission is performed. The detailed covert communication scheme is given in {Algorithm~\ref{JBJ_scheme}}.

\vspace{-7pt}
\begin{algorithm}[h]
\label{JBJ_scheme}
\caption{Covert Communication Scheme}
\begin{footnotesize}
\KwIn{Transmitted signal from Alice\;}
\KwOut{Covert rate at Alice\;}
{To initiate communication, Bob transmits pilots to Alice, and then Alice estimates the channels from Bob to Alice (i.e., ${\bf H}_{\mathrm{a},\mathrm{b}}$)}\; 
{By utilizing localization technology, Alice obtains the LoS CSI of the channels from Alice to each warden (i.e., $\overline{\mathbf{H}}_{\mathrm{a},\mathrm{w}_j}$)}\;
{Based on the CSI of the Alice–Bob channel and the LoS CSI of the channels from Alice to each warden, the beamforming vector is designed for covert rate maximization}\;
{Based on the beamforming vector and the default antenna orientation, Alice transmits covert signals at a rate no greater than the covert rate.}\;
{After receiving the signal from Alice, Bob decodes the covert information}\; 
\end{footnotesize}
\end{algorithm}
\vspace{-15pt}

\subsection{Performance Metric}
We adopt three metrics to evaluate the system performance, i.e., {detection error probability}, {transmission outage probability}, and {covert rate}.

1) {\bf Detection Error Probability (DEP)} is used to measure the detection performance at the warden, which is defined as the sum of the probability of false alarm and the probability of missed detection. Thus, the DEP of $j^{\mathrm{th}}$ warden can be expressed as
\begin{equation}
\footnotesize
\label{error}
    \xi_{j} = \mathbb{P}_{\mathrm{F}_j}+\mathbb{P}_{\mathrm{M}_j}.
\end{equation}
Note that Alice generally cannot obtain the exact value of the detection threshold $\tau_{j}$ set at the $j^{\mathrm{th}}$ warden. To ensure the robustness of the covert communication, we consider the worst case {where} each warden sets {the optimal threshold $\tau_{j}^{*}$ that minimizes} the DEP $\xi_{j}$ (denoted as $\xi_{j}^*$).
Thus, covert communication can be achieved as long as the covertness requirement is satisfied, i.e., $\xi_{j}^{*}\geq 1-\epsilon_{\mathrm{w}}, \forall j \in \{1, 2, \dots, n_\mathrm{w}\}$, where $\epsilon_{\mathrm{w}}$ is an arbitrarily small positive constant.

2) {\bf Transmission Outage Probability (TOP)} measures the reliability of the covert uplink transmission, and is defined as the probability that the transmit rate exceeds the channel capacity. We use $\zeta$ to denote TOP, which can be given by
\begin{equation}
\label{eq:TOP_def}
\footnotesize
    \zeta = \; {\rm Pr}\{ R_{\mathrm{a}} > C_{\mathrm{a},\mathrm{b}}\},
\end{equation}
where $R_{\mathrm{a}}$ denotes the transmit rate at Alice, $C_{\mathrm{a},\mathrm{b}} = \log_2 \left ( 1+\mathrm{SNR}_{\mathrm{a},\mathrm{b}} \right)$ represents the channel capacity of the Alice-Bob link.
$\mathrm{SNR}_{\mathrm{a},\mathrm{b}}$ is the signal noise ratio (SNR) at Bob, which can be determined as follows. The received signal at Bob is given by
\begin{equation}
\footnotesize
{y}_{\mathrm{b}}[k]= 
    \sqrt{F_{\mathrm{a},\mathrm{b}} G_{\mathrm{a},\mathrm{b}}G_{\mathrm{b},\mathrm{a}}P_{\mathrm{a}}}\mathbf{v}_{\mathrm{b}}^{\dagger}\mathbf{H}_{\mathrm{a},\mathrm{b}}\mathbf{w}_{\mathrm{a}} f_{\mathrm{a}}[k] \\+ \mathbf{v}_{\mathrm{b}}^{\dagger}\mathbf{n}_{\mathrm{b}}[k], 
    \label{Bob_recieve}
\end{equation}
where $\mathbf{n}_{\mathrm{b}}[k]\in \mathbb{C}^{M_{\mathrm{sat}} \times 1}$ is the additive white Gaussian noise at Bob, the elements of $\mathbf{n}_{\mathrm{b}}[k]$ follow the i.i.d. circularly symmetric complex Gaussian random distribution with zero mean and variance $\sigma_{\mathrm{b}}^{2}$, i.e., $\mathbf{n}_{\mathrm{b}}[k] \sim \mathcal{C N}\left(\mathbf{0}, \sigma_{\mathrm{b}}^2 \mathbf{I}_{M_{\mathrm{sat}}}\right)$. $\mathbf{v}_{\mathrm{b}} \in  \mathbb{C}^{M_{\mathrm{sat}} \times 1}$ is the LoS-EGC vector of Bob,  and $\mathbf{v}_{\mathrm{b}}=\mathbf{g}_{\mathrm{b},\mathrm{a}}/\sqrt{M_{\mathrm{sat}}}$. 
{Let $S_{\mathrm b} \triangleq{F_{\mathrm{a},\mathrm{b}} G_{\mathrm{a},\mathrm{b}}G_{\mathrm{b},\mathrm{a}} P_{\mathrm{a}}|\mathbf{v}_{\mathrm{b}}^{\dagger}\mathbf{H}_{\mathrm{a},\mathrm{b}} \mathbf{w}_{\mathrm{a}}|^2}$ denote the received signal power at Bob.} Then, ${\mathrm{SNR}}_{\mathrm{a},\mathrm{b}}$ can be given by
\begin{equation}
\footnotesize
    {\mathrm{SNR}}_{\mathrm{a},\mathrm{b}} = \frac{S_{\mathrm b}}{\sigma_{\mathrm{b}}^2}.
\end{equation}
To ensure the desired quality of service (QoS), the TOP of the covert uplink transmission should satisfy a tolerance threshold $\epsilon_{\mathrm{b}}$, i.e., $\zeta \leq \epsilon_{\mathrm{b}}$.

3) {\bf Covert Rate (CR)} is used to measure the performance of covert transmission from Alice to Bob, defined as the maximum transmit rate subject to the covertness requirement and the reliability constraints, which is given by 

\vspace{-2mm}
\begin{footnotesize}
\begin{subequations}
\label{rab}
    \begin{align}
      &R_{\mathrm{a}}^*=\max\; R_{\mathrm{a}} \\
     \text {s.t.}\quad & \xi_{j}^{*} \geq 1-\epsilon_\mathrm{w},\; \forall j \in \{1, 2, \dots, n_\mathrm{w}\}, \\
     & \zeta \leq \epsilon_{\mathrm{b}}.
    \end{align}
\end{subequations}
\end{footnotesize}

\section{Covert Communication Scheme under\\ Perfect Channel Estimation}
\label{sec:covert_scheme_perfect}

In this section, we study the covert communication scheme under perfect channel estimation for the Alice–Bob channel and the LoS channels from Alice to the $n_\mathrm{w}$ wardens.
First, we provide the theoretical results for DEP and TOP, which are associated with the covertness and reliability requirements, respectively.
Then, we explore the optimal beamforming design as well as the joint optimal beamforming and antenna orientation design for CR maximization.

\vspace{-15pt}
\subsection{Covertness Requirement under Perfect Channel Estimation}
\vspace{-3pt}
According to the definition of DEP in~(\ref{error}), we first need to determine the probability of false alarm  $\mathbb{P}_{{ \mathrm{F}_j}}$ and the probability of missed detection $\mathbb{P}_{{\mathrm{M}_j}}$. Based on the definition of $\mathbb{P}_{{ \mathrm{F}_j}}$, we can have
\begin{equation}
    \footnotesize
    \mathbb{P}_{{\rm F}_{j}}
    ={\rm Pr}\{\sigma^2_{\mathrm{w}_{j}} \geq \tau_j |\mathcal{H}_0 \}. 
\end{equation}
Following the PDF of $\sigma^2_{\mathrm{w}_{j}}$ given in \eqref{eq:noise_PDF}, the probability of false alarm $\mathbb{P}_{{\rm F}_{j}}$ can be given by

\vspace{-3mm}
\begin{footnotesize}
\begin{align}
    \mathbb{P}_{{\rm F}_{j}}
    &= \begin{cases} 1, & \tau_j < \hat{\sigma}_{\mathrm{w}_j}^{\mathrm{lb}}, \\ 
    \frac{\log\left(\hat{\sigma}_{\mathrm{w}_j}^{\mathrm{ub}}\right)-\log(\tau_j)}{2\log(\rho)}, &\hat{\sigma}_{\mathrm{w}_j}^{\mathrm{lb}} \leq \tau_j \leq \hat{\sigma}_{\mathrm{w}_j}^{\mathrm{ub}},\\
    0, &  \hat{\sigma}_{\mathrm{w}_j}^{\mathrm{ub}}<\tau_j. \end{cases} \label{eq:PF}
\end{align}    
\end{footnotesize}

We proceed to derive the probability of missed detection $\mathbb{P}_{{\mathrm{M}_j}}$. Based on its definition, we have
\begin{equation}
\footnotesize
    \begin{aligned}
        \mathbb{P}_{{\rm M}_{j}}
        &={\rm Pr}\{ S_{\mathrm{w}_j} + \sigma^2_{\mathrm{w}_{j}} \leq \tau_j |\mathcal{H}_0 \}.\\
    \end{aligned}
    \label{eq:PM_origin}
\end{equation}
Note that $\mathbb{P}_{{\rm M}_{j}}$ depends on two independent random variables, i.e., $S_{\mathrm{w}_j}$ and $\sigma^2_{\mathrm{w}_{j}}$. Before solving~\eqref{eq:PM_origin}, 
we first derive the PDF of $S_{\mathrm{w}_j}$ which is determined by $\mathbf{v}_{\mathrm{w}_j}^{\dagger}\mathbf{H}_{\mathrm{a},\mathrm{w}_{j}} \mathbf{w}_\mathrm{a}$.
Following~\eqref{eq:Rician_channel}, we have
\begin{equation}
\footnotesize\hspace{-1.5mm}\mathbf{v}_{\mathrm{w}_j}^{\dagger}\mathbf{H}_{\mathrm{a},\mathrm{w}_{j}}\mathbf{w}_\mathrm{a}  = \sqrt{\frac{K}{K+1}} \mathbf{v}_{\mathrm{w}_j}^{\dagger}\overline{\mathbf{H}}_{\mathrm{a},\mathrm{w}_j}\mathbf{w}_\mathrm{a}+\sqrt{\frac{1}{K+1}} \mathbf{v}_{\mathrm{w}_j}^{\dagger}\widetilde{\mathbf{H}}_{\mathrm{a},\mathrm{w}_j}\mathbf{w}_\mathrm{a}.
\end{equation}
Note that the randomness of $\mathbf{v}_{\mathrm{w}_j}^{\dagger}\mathbf{H}_{\mathrm{a},\mathrm{w}_{j}}\mathbf{w}_\mathrm{a}$ comes from $\widetilde{\mathbf{H}}_{\mathrm{a},\mathrm{w}_j}$. Since the design of $\mathbf{w}_\mathrm{a}$ is related to $\mathbf{H}_{\mathrm{a},\mathrm{b}}$ and $\overline{\mathbf{H}}_{\mathrm{a},\mathrm{w}_j}$ but independent of $\widetilde{\mathbf{H}}_{\mathrm{a},\mathrm{w}_j}$, the term $\mathbf{v}_{\mathrm{w}_j}^{\dagger}\overline{\mathbf{H}}_{\mathrm{a},\mathrm{w}_j}\mathbf{w}_\mathrm{a}$ is deterministic, and $\mathbf{v}_{\mathrm{w}_j}^{\dagger}\widetilde{\mathbf{H}}_{\mathrm{a},\mathrm{w}_j}\mathbf{w}_\mathrm{a}$ still follows the circularly symmetric complex Gaussian distribution. 
Thus, we can know from~[\citenum{Yan2016}, P. 4] that $\sqrt{S_{\mathrm{w}_j}}=\sqrt{F_{\mathrm{a},\mathrm{w}_{j}} G_{\mathrm{a},\mathrm{w}_{j}}G_{\mathrm{w}_j,\mathrm{a}} P_\mathrm{a}}|\mathbf{v}_{\mathrm{w}_j}^{\dagger}\mathbf{H}_{\mathrm{a},\mathrm{w}_{j}}\mathbf{w}_\mathrm{a}|$ follows the Rician distribution with Rician factor $\kappa_j$ and total power $\omega_j$, where $\kappa_j = {K|\mathbf{v}_{\mathrm{w}_j}^{\dagger} \overline{\mathbf{H}}_{\mathrm{a},\mathrm{w}_j} \mathbf{w}_\mathrm{a}|^2}$ and $\omega_j ={{F_{\mathrm{a},\mathrm{w}_{j}} G_{\mathrm{a},\mathrm{w}_{j}}G_{\mathrm{w}_{j},\mathrm{a}} P_\mathrm{a} }(1+ \kappa_j)}/{(1+K)}$. Then, the PDF of $\sqrt{S_{\mathrm{w}_j}}$ can be given by

\vspace{-3.5mm}
\begin{footnotesize}
 \begin{align}
        f_{\sqrt{S_{\mathrm{w}_j}}}(x)=&\frac{2(\kappa_j+1) x}{\omega_j} \exp \left(-\kappa_j-\frac{(\kappa_j+1) x^2}{\omega_j}\right) \nonumber\\
        &\times I_0\left(2 \sqrt{\frac{\kappa_j(\kappa_j+1)}{\omega_j}}x\right),
\end{align}    
\end{footnotesize}

\hspace{-3.5mm}where $I_0(\cdot)$ is the first-kind zero-order modified Bessel function. Since $I_0(\cdot)$ is expressed in an integral form, deriving a closed-form expression for $\sqrt{S_{\mathrm{w}_j}}$ is analytically intractable. To facilitate analysis, we interpret the Rician distribution as a special case of the Nakagami-m distribution~[\citenum{Yan2016}, Eq. (16)]. Thus, the PDF of $\sqrt{S_{\mathrm{w}_j}}$ can be expressed as 
\begin{equation}
\footnotesize
f_{\sqrt{S_{\mathrm{w}_j}}}(x)=\left(\frac{m_j}{\omega_j}\right)^{m_j} \frac{2x^{2 m_j-1}}{\Gamma(m_j)} \exp \left(-\frac{m_j}{\omega_j} x^2\right),
\end{equation}
where $\Gamma(a)=\int_0^{\infty} t^{a-1}e^{-t}  d t$ is the Gamma function. The shape parameter $m_j$ is defined as $m_j = {(\kappa_j + 1)^2}/({2\kappa_j + 1})$, and satisfies $m_j \geq 1$.
Then, the PDF of $S_{\mathrm{w}_j}$ can be given by

\vspace{-2mm}
\begin{footnotesize}
\begin{equation}
\label{eq:pdf_sw_j_i}
\begin{aligned}
    f_{S_{\mathrm{w}_j}}(x)
    &=f_{\sqrt{S_{\mathrm{w}_j}}}\left(\sqrt{x}\right)\left(\sqrt{x}\right)'\\
    &=\left(\frac{m_j}{\omega_j}\right)^{m_j} \frac{x^{ m_j-1}}{\Gamma(m_j)} \exp \left(-\frac{m_j}{\omega_j} x\right).
\end{aligned}
\end{equation}    
\end{footnotesize}

We now determine the cumulative distribution function (CDF) of $S_{\mathrm{w}_{j}} +\sigma^2_{\mathrm{w}_{j}}$. Let $Z_{j} \triangleq S_{\mathrm{w}_{j}} +\sigma^2_{\mathrm{w}_{j}}$, the CDF of $Z_{j}$ can be given by
\begin{equation}
\footnotesize
  \begin{aligned}
    F_{Z_j}(z)
        =&\;\mathrm{Pr}\{S_{\mathrm{w}_{j}} + \sigma^2_{\mathrm{w}_{j}}\leq z\}\\
        =&\int_{0}^{z}\int_{0}^{z-S_{\mathrm{w}_{j}}}f_{\sigma^2_{\mathrm{w}_{j}}}(\sigma^2_{\mathrm{w}_{j}})f_{S_{\mathrm{w}_{j}}}(S_{\mathrm{w}_{j}})d\sigma^2_{\mathrm{w}_{j}}dS_{\mathrm{w}_{j}}\\
        =&\int_{0}^{z-\hat{\sigma}_{\mathrm{w}_j}^{\mathrm{lb}}}\frac{\log(z-S_{\mathrm{w}_{j}})-\log\left(\hat{\sigma}_{\mathrm{w}_j}^{\mathrm{lb}}\right)}{2\log(\rho)}  \\
        &\times  \left(\frac{m_j}{\omega_j}\right)^{m_j}\frac{S_{\mathrm{w}_j}^{m_j-1}}{\Gamma(m_j)} \exp \left(-\frac{m_j}{\omega_j} S_{\mathrm{w}_{j}}\right) \, dS_{\mathrm{w}_{j}}.\label{F_z_w} 
\end{aligned}  
\end{equation}
Based on \eqref{eq:PM_origin} and \eqref{F_z_w}, the probability of missed detection $\mathbb{P}_{{\rm M}_{j}}$ can be given by

\vspace{-2mm}
\begin{footnotesize}
    \begin{align}
        \mathbb{P}_{{\rm M}_{j}}
        &=\begin{cases} 0, & \tau_j < \hat{\sigma}_{\mathrm{w}_j}^{\mathrm{lb}}, \\F_{Z_j}(\tau_j),  &  \hat{\sigma}_{\mathrm{w}_j}^{\mathrm{lb}} \le \tau_j. \end{cases}
        \label{eq:PM}
    \end{align}
\end{footnotesize}

Substituting \eqref{eq:PF} and \eqref{eq:PM} into ~(\ref{error}), we can obtain the DEP $\xi_j$ for the $j^{\mathrm{th}}$ warden, which is expressed as

\vspace{-2mm}
\begin{footnotesize}
 \begin{align}
 \label{DEP}
        \xi_j 
        =& \begin{cases} 1, & \tau_j < \hat{\sigma}_{\mathrm{w}_j}^{\mathrm{lb}}, \\  
        \frac{\log\left(\hat{\sigma}_{\mathrm{w}_j}^{\mathrm{ub}}\right)-\log(\tau_j)}{2\log(\rho)} + F_{Z_j}(\tau_j), &\hat{\sigma}_{\mathrm{w}_j}^{\mathrm{lb}} \le \tau_j \le \hat{\sigma}_{\mathrm{w}_j}^{\mathrm{ub}},\\  F_{Z_j}(\tau_j), &\hat{\sigma}_{\mathrm{w}_j}^{\mathrm{ub}}<\tau_j. \end{cases}
\end{align}
\end{footnotesize}It is noted that $F_{Z_j}(\tau_j)$ does not admit a closed-form expression, and is therefore unsuitable for further analysis of the DEP. 
To ensure the robustness of the covert communication scheme, we derive and analyze the lower bound of $F_{Z_j}(\tau_j)$, which is provided in the following lemma.

\begin{lemma}
\label{Theorem:F_z_lb}
   {In the considered system, the CDF of $Z_j$ is lower-bounded by $F_{Z_j}(\tau_j)\geq \check{F}_{Z_j}(\tau_j)$. Here, $\check{F}_{Z_j}(\tau_j)$ is given~by}

    \vspace{-2mm}
    \begin{footnotesize}
    \begin{equation}
        \begin{aligned}
           \check{F}_{Z_j}(\tau_j)\triangleq  \frac{\ \omega_j\log \big(\frac{\tau_j}{\hat{\sigma}_{\mathrm{w}_j}^{\mathrm{lb}}}\big) \left( \gamma \left(m_j+1,\mu_j(\tau_j)\right)- \mu_j(\tau_j) \gamma \left(m_j, \mu_j(\tau_j)\right)\right)}{2 \log (\rho) \Gamma (m_j+1) (\hat{\sigma}_{\mathrm{w}_j}^{\mathrm{lb}}- \tau_j)},           
            \label{eq:F_zw_lower_bound}
        \end{aligned}
    \end{equation}  
    \end{footnotesize}where $\Gamma(a,x)=\int_x^\infty t^{a-1}e^{-t}dt$ is the upper incomplete Gamma function, $\gamma(a,x)=\int_0^x t^{a-1}e^{-t}dt$ is the lower incomplete Gamma function, and $\mu_j(x)\triangleq{m_j \left(x-\hat{\sigma}_{\mathrm{w}_j}^{\mathrm{lb}}\right)}/{\omega_j}$.
\end{lemma}
\begin{proof}
    Please refer to Appendix \ref{Proof:F_z_lb}.
\end{proof}

According to Lemma~\ref{Theorem:F_z_lb}, the lower bound of the DEP $\xi_j$ for the $j^{\mathrm{th}}$ warden is denoted by $\check{\xi}_j$, and is given by

\vspace{-2mm}
\begin{footnotesize}
\begin{align}
    \xi_j \geq \check{\xi}_j& =  \begin{cases} 1, & \tau_j < \hat{\sigma}_{\mathrm{w}_j}^{\mathrm{lb}}, \\  
        \frac{\log\left(\hat{\sigma}_{\mathrm{w}_j}^{\mathrm{ub}}\right)-\log(\tau_j)}{2\log(\rho)} + \check{F}_{Z_j}(\tau_j), &\hat{\sigma}_{\mathrm{w}_j}^{\mathrm{lb}} \le \tau_j \le \hat{\sigma}_{\mathrm{w}_j}^{\mathrm{ub}},\\  \check{F}_{Z_j}(\tau_j), &\hat{\sigma}_{\mathrm{w}_j}^{\mathrm{ub}}<\tau_j. \end{cases}
\end{align}    
\end{footnotesize}Note that the detection threshold $\tau_j$ set by the $j^{\mathrm{th}}$ warden is unknown to Alice. To ensure the effectiveness of the covert communication scheme, we consider the worst-case scenario where the $j^{\mathrm{th}}$ warden adopts the optimal detection threshold $\tau_j^{*}$ to minimize $\check{\xi}_j$. To this end, we further identify $\tau_j^{*}$, which is given in the following lemma.


\begin{lemma}
\label{Theorem:Opt_tau}
{The optimal detection threshold of the $j^{\mathrm{th}}$ warden is determined as $\tau_j^* = \min(\tau_j',\hat{\sigma}_{\mathrm{w}_j}^{\mathrm{ub}})$, here $\tau_{j}^{'}$ satisfies the following condition}
\begin{equation}
\footnotesize
    \begin{aligned}
          \mu_j(\tau_j') \Gamma \left(m_j,\mu_j(\tau_j')\right)-  \gamma \left(m_j+1,\mu_j(\tau_j')\right) \nu_j(\tau_j')=0,
    \end{aligned}
    \label{lemma1_equation}
\end{equation}
where $\nu_j(x)\triangleq \frac{x \log \left({x}/{\hat{\sigma}_{\mathrm{w}_j}^{\mathrm{lb}}}\right)}{x-\hat{\sigma}_{\mathrm{w}_j}^{\mathrm{lb}}}-1$.


\end{lemma}
\begin{proof}
    Please refer to Appendix \ref{Proof:Opt_tau}.
\end{proof}

By substituting $\tau_j^*$ into $\check{\xi}_{j}$, we can obtain the minimum $\check{\xi}_{j}$, denoted by $\check{\xi}_{j}^*$, which can be given by
\begin{equation}
    \footnotesize
    \check{\xi}_{j}^* = 1-\frac{\omega_j\log \big(\frac{ \tau_j^*}{\hat{\sigma}_{\mathrm{w}_j}^{\mathrm{lb}}}\big) \left(\mu_j(\tau_j^*) \Gamma \left(m_j,\mu_j(\tau_j^*)\right)+ \gamma \left(m_j+1,\mu_j\left(\tau_j^*\right)\right)\right)}{2  \log (\rho) \Gamma (m_j+1) ( \tau_j^*-\hat{\sigma}_{\mathrm{w}_j}^{\mathrm{lb}})}. \label{eq:xi_star}
\end{equation}
It can be observed from~\eqref{lemma1_equation} that a closed-form expression for $\tau_j^*$ is not available. As a result, the closed-form of $\check{\xi}_j^*$ is also intractable, {which poses substantial challenges in analyzing the covertness constraint and designing optimal beamforming.} To address this issue, we further derive a lower bound of $\check{\xi}_j^*$, which is given in the following theorem.
%

\begin{theorem}
\label{Theorem:DEP_lower_bound}
   { The lower bound of $\check{\xi}_{j}^*$ is decided as $\check{\xi}_{j}^*>\check{\xi}^{\mathrm{lb}}_{j}$, where} 
    \begin{equation}
\label{eq:DEP_lower_bound}
\footnotesize
      \check{\xi}^{\mathrm{lb}}_{j}\triangleq 1-\frac{ \omega_j }{2\log(\rho)\hat{\sigma}_{\mathrm{w}_j}^{\mathrm{lb}}}.
\end{equation}
\end{theorem}
\begin{proof}
   Please refer to Appendix \ref{Proof:DEP_lower_bound}.
\end{proof}
We consider $\check{\xi}^{\mathrm{lb}}_{j}\geq 1-\epsilon_\mathrm{w},\, \forall j \in \{1, 2, \dots, n_\mathrm{w}\}$ as the covertness requirement, which is the constraint that must be satisfied for achieving robust covert communication.

\vspace{-7pt}
\subsection{Reliability Requirement under Perfect Channel Estimation}
\vspace{-2pt}
Due to noise uncertainty at Bob, transmission outages may occur on the Alice–Bob link.
According to the definition of TOP in  \eqref{eq:TOP_def},  we have
\begin{equation}
\footnotesize
 \begin{aligned}
    \zeta =& \; {\rm Pr}\{ R_{\mathrm{a}} > C_{\mathrm{a},\mathrm{b}}\}  \\
    =& \begin{cases}
         0, & \hat{\sigma}_{\mathrm{b}}^{\mathrm{ub}} \leq \frac{S_{\mathrm b}}{(2^{R_{\mathrm{a}}} -1)},\\
         \frac{\log\left({\hat{\sigma}_{\mathrm{b}}^{\mathrm{ub}}}\right)+\log(2^{R_{\mathrm{a}}} -1)-\log(S_{\mathrm b})}{2\log(\rho)}, &  \hat{\sigma}_{\mathrm{b}}^{\mathrm{ub}} >  \frac{S_{\mathrm b}}{(2^{R_{\mathrm{a}}} -1)}.
        \end{cases}\label{eq:TOP}
\end{aligned}     
\end{equation}

The reliability requirement is guaranteed by limiting the TOP $\zeta$ to stay below a tolerable threshold $\epsilon_{\mathrm{b}}$, i.e., $\zeta \leq \epsilon_{\mathrm{b}}$, thereby ensuring QoS of the Alice–Bob transmission.

\vspace{-7pt}
\subsection{Beamforming Design under Perfect Channel Estimation}
\vspace{-2pt}
We now investigate the optimal beamforming design for CR maximization under perfect channel estimation.
Our goal is to jointly optimize the beamforming vector, transmit power and transmit rate of Alice to achieve the maximum CR of the satellite-terrestrial covert uplink transmission. 
Such an optimization problem can be mathematically formulated as

\vspace{-3mm}
\begin{footnotesize}
\begin{subequations}
    \label{Prob:Beamforming_perfect}
    \begin{align}
        & \max _{ \mathbf{w}_{\mathrm{a}}, P_\mathrm{a}, R_{\mathrm{a}}}\; R_{\mathrm{a}} \label{JBJ_opt}\\
       \text { s.t. }
        &\, \left \|\mathbf{w}_{\mathrm{a}}\right \|^2 =1,\label{wa_normal_cons} \\
        & \;P_\mathrm{a} \leq P_{\mathrm{a},\max}, \label{alice_power_cons}\\
        & \;\check{\xi}^{\mathrm{lb}}_{j} \geq 1-\epsilon_\mathrm{w},\, \forall j \in \{1, 2, \dots, n_\mathrm{w}\}, \label{covert_req}\\
        &\; \zeta \leq \epsilon_{\mathrm{b}} \label{cons:Pout_req},
    \end{align}
\end{subequations}
\end{footnotesize}where $P_{\mathrm{a},\max}$ is the maximum transmit power at Alice.
Constraints \eqref{covert_req} and \eqref{cons:Pout_req} represent the covertness and reliability requirements under perfect channel estimation, respectively.

Following \eqref{eq:DEP_lower_bound} and \eqref{eq:TOP}, problem~\eqref{Prob:Beamforming_perfect} can be reformulated~as 

\vspace{-3mm}
\begin{subequations}
    \label{Prob:Beamforming_perfect_simplify}
    \footnotesize
    \begin{align}
        & \max _{ \mathbf{w}_{\mathrm{a}}, P_\mathrm{a}, r_{\mathrm{a}}}\; r_{\mathrm{a}} \label{JBJ_opt}\\
       \text { s.t. }
        & \;{ G_{\mathrm{a},\mathrm{w}_j}P_\mathrm{a}(1 + K|\mathbf{v}_{\mathrm{w}_j}^{\dagger} \overline{\mathbf{H}}_{\mathrm{a},\mathrm{w}_j} \mathbf{w}_{\mathrm{a}}|^2})  \leq \eta_{\mathrm{w}_j},\nonumber\\&\; \forall j \in \{1, 2, \dots, n_\mathrm{w}\}, \label{covert_req_simplify}\\
        &\; { P_\mathrm{a}\eta_\mathrm{b}|\mathbf{v}_{\mathrm{b}}^{\dagger}\mathbf{H}_{\mathrm{a},\mathrm{b}} \mathbf{w}_\mathrm{a}|^2}   \geq r_{\mathrm{a}} \label{cons:Pout_req_simplify},\\
        &\; \text{\eqref{wa_normal_cons}, \eqref{alice_power_cons},}
    \end{align}
\end{subequations}
where $r_{\mathrm{a}} = 2^{R_{\mathrm{a}}}-1$, $\eta_{\mathrm{w}_j} =\; $\scalebox{0.92}{${2 \epsilon_\mathrm{w}\log(\rho)(1+K)\hat{\sigma}_{\mathrm{w}_j}^{\mathrm{lb}}}/{( F_{\mathrm{a},\mathrm{w}_j}G_{\mathrm{w}_j,\mathrm{a}}})$} and $\eta_\mathrm{b} = {\rho^{2 \epsilon_{\mathrm{b}}-1}F_{\mathrm{a},\mathrm{b}} G_{\mathrm{a},\mathrm{b}}G_{\mathrm{b},\mathrm{a}} }/{\hat {\sigma}_{\mathrm{b}}}$.
Let $\widehat{\mathbf{W}}_{\mathrm{a}} =P_\mathrm{a}\mathbf{w}_{\mathrm{a}} \mathbf{w}_{\mathrm{a}}^{\dagger}$, $\mathbf{V}_{\mathrm{b}}=\mathbf{v}_{\mathrm{b}}\mathbf{v}_{\mathrm{b}}^{\dagger}$ and $\mathbf{V}_{\mathrm{w}_j}=\mathbf{v}_{\mathrm{w}_j}\mathbf{v}_{\mathrm{w}_j}^{\dagger}$, problem \eqref{Prob:Beamforming_perfect_simplify} can be reformulated as

\vspace{-1mm}
\begin{footnotesize}
\begin{subequations}

    \label{Prob:Beamforming_perfect_SDR}
    \begin{align}
        & \max _{ \widehat{\mathbf{W}}_{\mathrm{a}}}\;\mathrm{Tr}(\mathbf{V}_{\mathrm{b}}\mathbf{H}_{\mathrm{a},\mathrm{b}}\widehat{\mathbf{W}}_{\mathrm{a}}\mathbf{H}_{\mathrm{a},\mathrm{b}}^\dagger) \\
       \text { s.t. }
        & \; G_{\mathrm{a},\mathrm{w}_j}\left(\mathrm{Tr}(\widehat{\mathbf{W}}_{\mathrm{a}})+K\,\mathrm{Tr}(\mathbf{V}_{\mathrm{w}_j}\overline{\mathbf{H}}_{\mathrm{a},\mathrm{w}_j}\widehat{\mathbf{W}}_{\mathrm{a}}\overline{\mathbf{H}}_{\mathrm{a},\mathrm{w}_j}^{\dagger})\right)\leq \eta_{\mathrm{w}_j}, \nonumber \\ & \; \forall j \in \{1, 2, \dots, n_\mathrm{w}\}, \label{covert_req_simplfy}\\
        &\; \mathrm{Tr}\left(\widehat{\mathbf{W}}_{\mathrm{a}}\right) \leq P_{\mathrm{a},\max}, \; \label{wa_power_cons_SDR}\\
        &\; \widehat{\mathbf{W}}_{\mathrm{a}} \succeq \mathbf{0},  \label{SD_cons}\\
        & \; \operatorname{Rank}\left(\widehat{\mathbf{W}}_{\mathrm{a}}\right)=1. \; \label{SD_cons2}
    \end{align}
\end{subequations}
\end{footnotesize}Note that constraint~(\ref{SD_cons2}) is non-convex, making problem \eqref{Prob:Beamforming_perfect_SDR} difficult to solve. To address this issue, we apply the semi-definite relaxation (SDR) technique \cite{Luo2010} to relax constraint (\ref{SD_cons2}) and  transform problem~\eqref{Prob:Beamforming_perfect_SDR} into a semi-definite programming (SDP) problem, which can be given by

\begin{subequations}
\begin{footnotesize}
    \label{Prob:Beamforming_perfect_SDP}
    \begin{align}
        & \max _{ \widehat{\mathbf{W}}_{\mathrm{a}}}\;\mathrm{Tr}(\mathbf{V}_{\mathrm{b}}\mathbf{H}_{\mathrm{a},\mathrm{b}}\widehat{\mathbf{W}}_{\mathrm{a}}\mathbf{H}_{\mathrm{a},\mathrm{b}}^\dagger) \\
       \text { s.t. }
        &\;\text{\eqref{covert_req_simplfy}, \eqref{wa_power_cons_SDR}, \eqref{SD_cons}.}
    \end{align}
\end{footnotesize}
\end{subequations}Therefore, problem~\eqref{Prob:Beamforming_perfect_SDP} can be solved using a standard convex optimization solver, and the optimal solution is denoted by $\widehat{\mathbf{W}}_{\mathrm{a}}^{*}$.
Due to the relaxation of constraint \eqref{SD_cons2}, the optimal solution $\widehat{\mathbf{W}}_{\mathrm{a}}^{*}$ is generally not a rank-one matrix. 
Therefore, after first acquiring the solution $\widehat{\mathbf{W}}_{\mathrm{a}}^{*}$ by solving problem~\eqref{Prob:Beamforming_perfect_SDP}, the high-quality beamforming vector $\hat{\mathbf{w}}_{\mathrm{a}}^{*}$ can be obtained via the widely employed Gaussian randomization procedure~\cite{Luo2010,Xiang2012,Ma2021}.
Given $\hat{\mathbf{w}}_{\mathrm{a}}^{*}$, the optimal beamforming vector and optimal transmit power of Alice can be obtained by $ \mathbf{w}_{\mathrm{a}}^* = \hat{\mathbf{w}}_{\mathrm{a}}^{*}/\|\hat{\mathbf{w}}_{\mathrm{a}}^{*}\|$ and $P^*_{\mathrm{a}} = \|\hat{\mathbf{w}}_{\mathrm{a}}^{*}\|^2$, respectively.
The overall procedure for solving problem~\eqref{Prob:Beamforming_perfect} is summarized in Algorithm~\ref{algorithm:Opt-JBJ}.

\begin{algorithm}[ht]
\begin{footnotesize}
\label{algorithm:Opt-JBJ}
\caption{Optimal Beamforming Design}
\KwIn{Channel coefficient matrices and maximum transmit power\;}
\KwOut{Optimal beamforming vector $\mathbf{w}^*_{\mathrm{a}}$, \\ \hspace{11mm}optimal transmit power $P_\mathrm{a}^*$, maximum CR $R_{\mathrm{a}}^*$  \;}
{Obtain $\widehat{\mathbf{W}}_{\mathrm{a}}^{*}$ by solving SDP problem \eqref{Prob:Beamforming_perfect_SDP}}\;
{Using the Gaussian randomization procedure to obtain $\hat{\mathbf{w}}_{\mathrm{a}}^{*}$}\;
{Obtain optimal beamforming vector ${\mathbf{w}}_{\mathrm{a}}^*\leftarrow\hat{\mathbf{w}}_{\mathrm{a}}^{*}/\|\hat{\mathbf{w}}_{\mathrm{a}}^{*}\|$}\;
{Obtain optimal transmit power ${P}_{\mathrm{a}}^*\leftarrow\|\hat{\mathbf{w}}_{\mathrm{a}}^{*}\|^2$}\;
{Obtain maximum CR $R_{\mathrm{a}}^*$ based on \eqref{cons:Pout_req_simplify}}\;
\Return {$\mathbf{w}^*_{\mathrm{a}}$, $P_\mathrm{a}^*$, \rm {and} $R_{\mathrm{a}}^*$ };
\end{footnotesize}
\end{algorithm}

\subsection{Joint Beamforming and Antenna Orientation Design under Perfect Channel Estimation }

It is worth noting that the optimal antenna orientation configuration can significantly enhance the CR~\cite{Forouzesh2020a}.
Therefore, this subsection further investigates the joint design of beamforming and antenna orientation to maximize the CR under the perfect channel estimation.

\vspace{-1mm}
\begin{footnotesize}
\begin{subequations}
    \label{Prob:Beamforming_antenna_perfect}
    \begin{align}
        & \max _{ \mathbf{o}_{\mathrm{a}}, \mathbf{w}_{\mathrm{a}}, P_\mathrm{a}, R_{\mathrm{a}}}\; R_{\mathrm{a}} \label{JBJ-B_opt}\\
       \text { s.t. }
        &\;\text{\eqref{wa_normal_cons}, \eqref{alice_power_cons}, \eqref{covert_req}, \eqref{cons:Pout_req}.}
    \end{align}
\end{subequations}
\end{footnotesize}Similar to problem~\eqref{Prob:Beamforming_perfect}, we first simplify problem~\eqref{Prob:Beamforming_antenna_perfect} and then apply the SDR technique to reformulate it as an SDP problem. By introducing $\mathbf{W}_{\mathrm{a}} = \mathbf{w}_{\mathrm{a}} \mathbf{w}_{\mathrm{a}}^\dagger$, problem~\eqref{Prob:Beamforming_antenna_perfect} can be

\newpage
\hspace{-3.5mm} rewritten as

\begin{subequations}
\begin{footnotesize}
    \label{Prob:antenna_perfect}
    \begin{align}
        & \max _{ \mathbf{o}_{\mathrm{a}}, \mathbf{W}_\mathrm{a}, P_\mathrm{a}}\; G_{\mathrm{a},\mathrm{b}} P_\mathrm{a} \mathrm{Tr}(\mathbf{V}_{\mathrm{b}}\mathbf{H}_{\mathrm{a},\mathrm{b}}\mathbf{W}_\mathrm{a}\mathbf{H}_{\mathrm{a},\mathrm{b}}^\dagger) \label{obj:antenna_perfect}\\
       \text { s.t. }
       & \;{G_{\mathrm{a},\mathrm{w}_j}P_\mathrm{a}\left(1+K\,\mathrm{Tr}(\mathbf{V}_{\mathrm{w}_j}\overline{\mathbf{H}}_{\mathrm{a},\mathrm{w}_j}\mathbf{W}_\mathrm{a}\overline{\mathbf{H}}_{\mathrm{a},\mathrm{w}_j}^{\dagger})\right)}\leq \eta_{\mathrm{w}_j}, \nonumber \\ & \; \forall j \in \{1, 2, \dots, n_\mathrm{w}\}, \label{cons:covert_req_antenna}\\
       &\; \mathrm{Tr}\left(\mathbf{W}_{\mathrm{a}}\right) =1, \label{cons:Tr_Wa_1}\\
       &\; \mathbf{W}_\mathrm{a} \succeq \mathbf{0},  \label{SD_cons_Wa}\\
       &\;\text{\eqref{alice_power_cons}.}
    \end{align}
    \end{footnotesize}
\end{subequations}Note that problem \eqref{Prob:antenna_perfect} involves strong coupling among optimization variables $\mathbf{o}_{\mathrm{a}}$, $\mathbf{W}_{\mathrm{a}}$, and $P_\mathrm{a}$, which significantly increases the solution complexity.
To efficiently solve this optimization problem, we decompose problem \eqref{Prob:antenna_perfect} into two subproblems, and then utilize an alternating optimization approach to obtain the local optimal solution of problem \eqref{Prob:antenna_perfect}.
Particularly, the first subproblem deals with the optimization of the beamforming matrix $\mathbf{W}_{\mathrm{a}}$ under the given transmit power $P_\mathrm{a}$ and antenna boresight vector $\mathbf{o}_{\mathrm{a}}$, while the second subproblem concerns the joint optimization of $\mathbf{o}_{\mathrm{a}}$ and $P_\mathrm{a}$ under the given $\mathbf{W}_{\mathrm{a}}$.
The first subproblem is consequently formulated as

\vspace{-3.5mm}
\begin{footnotesize}
\begin{subequations}
    \label{Prob:antenna_perfect_wa}
    \begin{align}
        & \max_{\mathbf{W}_\mathrm{a}}\;  \mathrm{Tr}(\mathbf{V}_{\mathrm{b}}\mathbf{H}_{\mathrm{a},\mathrm{b}}\mathbf{W}_\mathrm{a}\mathbf{H}_{\mathrm{a},\mathrm{b}}^\dagger) \\
       \text { s.t. }
        &\;\text{ \eqref{alice_power_cons}, \eqref{cons:covert_req_antenna}, \eqref{cons:Tr_Wa_1}, \eqref{SD_cons_Wa}.}
    \end{align}
\end{subequations}
\end{footnotesize}Similar to problem \eqref{Prob:Beamforming_perfect_SDR}, problem \eqref{Prob:antenna_perfect_wa} is an SDP problem and can be efficiently solved using a standard convex optimization solver, with the optimal beamforming matrix denoted by $\mathbf{W}_\mathrm{a}^*$.

With the given $\mathbf{W}_\mathrm{a}^*$, the second subproblem can be formulated as

\vspace{-4.5mm}
\begin{footnotesize}
\begin{subequations}
    \label{Prob:antenna_perfect_Gab}
    \begin{align}
        & \max _{ \mathbf{o}_{\mathrm{a}}, P_\mathrm{a}}\; G_{\mathrm{a},\mathrm{b}} P_\mathrm{a} \label{obj:Gab_Pa_wa} \\
       \text { s.t. }
        &\text{ \eqref{alice_power_cons}, \eqref{cons:covert_req_antenna}.}
    \end{align}
\end{subequations}
\end{footnotesize}Since the variables $\mathbf{o}_{\mathrm{a}}$ and $P_\mathrm{a}$ are still coupled, it is also difficult to solve the problem \eqref{Prob:antenna_perfect_Gab}. To solve this issue, we first explore the necessary condition that the optimal antenna boresight vector $\mathbf{o}_{\mathrm{a}}^*$ should satisfy, which is summarized in the following theorem.

\begin{theorem}
\vspace{-1mm}
\label{Theorem:minimum_psi_ab}
    The necessary condition of the optimal antenna boresight vector $\mathbf{o}_{\mathrm{a}}^*$ is that the off-boresight angle of Bob should satisfy $\vartheta_{\mathrm{a},\mathrm{b}} = \vartheta_0$. 
\end{theorem}
\begin{proof}
\vspace{-1mm}
    Please refer to Appendix \ref{Proof:minimum_psi_ab}.
\end{proof}
\vspace{-0.5mm}

We then explore the optimal antenna boresight vector based on Theorem \ref{Theorem:minimum_psi_ab} and Rodrigues' rotation formula.
Accordingly, $\mathbf{o}_{\mathrm{a}}$ can be formulated by
\begin{equation}
\label{oa}
\footnotesize
    \mathbf{o}_{\mathrm{a}} =  \mathbf{v} \cos \theta+(\mathbf{u} \times \mathbf{v}) \sin \theta+\mathbf{u}(\mathbf{u} \cdot \mathbf{v})(1-\cos \theta),
\end{equation}
where $\theta \in [0, 2\pi]$ denotes the angle between the vector $\mathbf{o}_{\mathrm{a},\mathrm{w}_j} \triangleq \mathbf{o}_{\mathrm{a}} - \mathbf{u}_{\mathrm{a},\mathrm{w}_j}$ and the $x$–$y$ plane of the Cartesian coordinate system, $\mathbf{u} = \frac{\mathbf{u}_{\mathrm{a},\mathrm{b}}}{\|\mathbf{u}_{\mathrm{a},\mathrm{b}}\|}$, $\mathbf{v} =\cos \left(\vartheta_0\right) \mathbf{u}+\sin \left(\vartheta_0\right) \mathbf{v}_{\perp}$, $\mathbf{v}_{\perp}=\frac{\mathbf{e}-(\mathbf{e} \cdot \mathbf{u}) \mathbf{u}}{\|\mathbf{e}-(\mathbf{e} \cdot \mathbf{u}) \mathbf{u}\|}$, and $\mathbf{e}=[0,0,1]^T$.
Note that for a given $\mathbf{o}_{\mathrm{a}}$, we can obtain a $P_\mathrm{a}$ by solving problem~\eqref{Prob:antenna_perfect_Gab}. To determine the optimal solution $(\mathbf{o}_{\mathrm{a}}, P_{\mathrm{a}}^*)$ that maximizes the objective function~\eqref{obj:Gab_Pa_wa}, we first employ the 1-D search method to generate a set of $\mathbf{o}_{\mathrm{a}}$ values based on~\eqref{oa}. Subsequently, each $\mathbf{o}_{\mathrm{a}}$ is substituted into problem~\eqref{Prob:antenna_perfect_Gab} to compute the corresponding $P_{\mathrm{a}}$ that maximizes~\eqref{obj:Gab_Pa_wa}. By exhaustively evaluating all candidate $\mathbf{o}_{\mathrm{a}}$, the optimal solution $(\mathbf{o}_{\mathrm{a}}, P_{\mathrm{a}}^*)$ and the corresponding maximum value of~\eqref{obj:Gab_Pa_wa} can be obtained.
To obtain the optimal solution and the maximum $R_{\mathrm{a}}^*$ of problem~\eqref{Prob:Beamforming_antenna_perfect}, an alternating optimization approach is employed between~\eqref{Prob:antenna_perfect_wa} and~\eqref{Prob:antenna_perfect_Gab} until convergence. The details for identifying the optimal solution to problem \eqref{Prob:Beamforming_antenna_perfect} are summarized in Algorithm \ref{overall_algorithm_P1}.

\vspace{-7pt}
\begin{algorithm}[hbt]
\label{overall_algorithm_P1}
\begin{footnotesize}
\caption{Joint Optimal Design of Beamforming and Antenna Orientation }
\KwIn{Channel coefficient matrices and maximum transmit power\;}
\KwOut{Optimal beamforming vectors $\mathbf{w}^*_{\mathrm{a}}$, optimal antenna boresight vector $\mathbf{o}_{\mathrm{a}}^*$, optimal transmit power $P_\mathrm{a}^*$,\\ \hspace{9.8mm} maximum CR $R_{\mathrm{a}}^*$\;}
{Set the maximum number of iterations $I$, the maximum tolerance $\varepsilon$ and the iteration index $i \gets 1$, initialize $\mathbf{o}_{\mathrm{a}}^{(0)}=\mathbf{u}_{\mathrm{a},\mathrm{b}}$, $ P_\mathrm{a}^{(0)}=0$}\;
\Repeat{ $\left |g\left(\mathbf{o}_{\mathrm{a}}^{(i)},\mathbf{W}^{*,(i)}_{\mathrm{a}},P_\mathrm{a}^{(i)}\right)-g\left(\mathbf{o}_{\mathrm{a}}^{(i-1)},\mathbf{W}^{*,(i-1)}_{\mathrm{a}},P_\mathrm{a}^{(i-1)}\right)\right |\leq \varepsilon$ or $i\ge I$}
{{Update $\mathbf{W}_{\mathrm{a}}^{*,(i)}$ by solving SDP problem \eqref{Prob:antenna_perfect_wa}}\;
{ Update $\mathbf{o}_{\mathrm{a}}^{(i)}$ and $P_\mathrm{a}^{(i)}$ by solving problem \eqref{Prob:antenna_perfect_Gab}\;
  $i \gets i+1$\;}
}
{$\mathbf{W}_{\mathrm{a}}^* \leftarrow \mathbf{W}_{\mathrm{a}}^{*,(i)}$, $\mathbf{o}_{\mathrm{a}}^* \leftarrow \mathbf{o}_{\mathrm{a}}^{(i)}$, $P_\mathrm{a}^* \leftarrow P_\mathrm{a}^{(i)}$}\;
{Apply the Gaussian randomization procedure to obtain $\mathbf{w}_{\mathrm{a}}^{*}$}\;
{Obtain maximum CR $R_{\mathrm{a}}^*$ based on \eqref{cons:Pout_req_simplify}}\;
\Return {$\mathbf{w}^*_{\mathrm{a}}$, $\mathbf{o}_{\mathrm{a}}^{*}$, $P_\mathrm{a}^*$, $R_{\mathrm{a}}^*$}\;
\end{footnotesize}
\end{algorithm}
\vspace{-7pt}

\vspace{-9pt}
\section{Covert Communication Scheme under Imperfect Channel Estimation}
\label{sec:covert_scheme_imp}
\vspace{-2pt}
In this section, we study the covert communication scheme where the channel from Alice to Bob and the LoS channels from Alice to the $n_\mathrm{w}$ wardens are imperfectly estimated by Alice. 
We first provide the channel estimation error model, and then analyze the covertness and reliability requirements under imperfect channel estimation. Subsequently, CR maximization problems are explored to identify the optimal beamforming design as well as the joint optimal beamforming and antenna orientation design under imperfect channel estimation.

\vspace{-10pt}
\subsection{Channel Estimation Error Model}
\vspace{-2pt}
{Due to the effects of imperfect synchronization, signal delays, and frequency offsets during channel estimation, it is challenging for Alice to perfectly estimate the channel from Alice to Bob and the LoS channels from Alice to each of $n_\mathrm{w}$ wardens~\cite{Li2022}. {Considering a worst-case scenario for robust beamforming designs,} we adopt the bounded channel estimation error model to characterize the channel estimation errors in these channels~[\citenum{Ma2021}, P. 4], which can be given by}

\vspace{-3.5mm}
\begin{footnotesize}
\begin{align}
    \mathbf{H}_{\mathrm{a}, \mathrm{b}} &= {\mathbf{H}}_{\mathrm{a}, \mathrm{b}}' + \Delta \mathbf{H}_{\mathrm{a}, \mathrm{b}}, \label{imperfect_channel_ab}\\
    \overline{\mathbf{H}}_{\mathrm{a},\mathrm{w}_j} &=  \overline{\mathbf{H}}_{\mathrm{a},\mathrm{w}_{j}}' + \Delta \overline{\mathbf{H}}_{\mathrm{a},\mathrm{w}_j}, \;  j \in \{1, 2, \dots, n_\mathrm{w}\}.
\end{align}   
\end{footnotesize}Here, ${\mathbf{H}}_{\mathrm{a},\mathrm{b}}'$ and $\overline{\mathbf{H}}_{\mathrm{a},\mathrm{w}_{j}}'$ denote the estimated CSI of the Alice-Bob channel and that of the LoS channel from Alice to the $j^{\mathrm{th}}$ warden, respectively. $\Delta \mathbf{H}_{\mathrm{a},\mathrm{b}}$ and $\Delta \overline{\mathbf{H}}_{\mathrm{a},\mathrm{w}_j}$ represent the corresponding channel estimation error matrices, which are characterized by ellipsoidal regions and can be given by

\vspace{-3.5mm}
\begin{footnotesize}
 \begin{align}
    \Omega_{\mathrm{a},\mathrm{b}} \triangleq & \left\{\Delta \mathbf{H}_{\mathrm{a},\mathrm{b}}| \; \|\Delta \mathbf{H}_{\mathrm{a},\mathrm{b}}\right\|^2\le  \delta_{\mathrm{b}}\}, \label{imperfect_channel_ab_space} \\
    \Omega_{\mathrm{a},\mathrm{w}_j} \triangleq & \left\{\Delta \overline{\mathbf{H}}_{\mathrm{a},\mathrm{w}_j}|\; \|\Delta \overline{\mathbf{H}}_{\mathrm{a},\mathrm{w}_j} \|^2\le  \delta_{\mathrm{w}_j}\right\},
    \; j \in \{1, 2, \dots, n_\mathrm{w}\},
\end{align}   
\end{footnotesize}where $\delta_{\mathrm{b}}$ and $\delta_{\mathrm{w}_j}$ are the radii of the uncertainty region.

\vspace{-9pt}
\subsection{\hspace{-5pt}Covertness \hspace{-0.5pt}Requirement \hspace{-0.5pt}under \hspace{-0.5pt}Imperfect \hspace{-0.5pt}Channel \hspace{-0.5pt}Estimation}
\vspace{-2pt}
To ensure the robustness of covert communication under imperfect channel estimation, we consider that the covertness requirement~\eqref{covert_req_simplify} can be satisfied under all possible channel errors. Thus, the covertness requirement under imperfect channel estimation is given by
\begin{equation}
    \footnotesize 
    \begin{aligned}
            &\max_{\Delta \overline{\mathbf{H}}_{\mathrm{a},\mathrm{w}_j} \in \Omega_{\mathrm{a},\mathrm{w}_j}} P_\mathrm{a}(1 + K|\mathbf{v}_{\mathrm{w}_j}^{\dagger} \overline{\mathbf{H}}_{\mathrm{a},\mathrm{w}_j} \mathbf{w}_{\mathrm{a}}|^2)\leq \eta_{\mathrm{w}_j}, \forall j \in \{1, 2, \dots, n_\mathrm{w}\}. 
    \end{aligned}
    \label{cons:imperfect_DEP_1}
\end{equation}
Note that \eqref{cons:imperfect_DEP_1} is a semi-infinite constraint (SIC), which is intractable.
By using the triangle inequality and the Cauchy-Schwarz inequality [\citenum{Xiang2012}, P. 2], we can have

\vspace{-3.5mm}
\begin{footnotesize}
\begin{align}
    &|\mathbf{v}_{\mathrm{w}_j}^{\dagger} (\overline{\mathbf{H}}_{\mathrm{a},\mathrm{w}_{j}} + \Delta \overline{\mathbf{H}}_{\mathrm{a},\mathrm{w}_j}) \mathbf{w}_{\mathrm{a}}|\nonumber\\ &\le  |\mathbf{v}_{\mathrm{w}_j}^{\dagger} \overline{\mathbf{H}}_{\mathrm{a},\mathrm{w}_{j}}  \mathbf{w}_{\mathrm{a}}|+|\mathbf{v}_{\mathrm{w}_j}^{\dagger} \Delta \overline{\mathbf{H}}_{\mathrm{a},\mathrm{w}_j}  \mathbf{w}_{\mathrm{a}}|\nonumber\\
    &\le  |\mathbf{v}_{\mathrm{w}_j}^{\dagger} \overline{\mathbf{H}}_{\mathrm{a},\mathrm{w}_{j}}  \mathbf{w}_{\mathrm{a}}|+\|\mathbf{v}_{\mathrm{w}_j}\| \cdot \|\Delta \overline{\mathbf{H}}_{\mathrm{a},\mathrm{w}_j}\|\cdot \|  \mathbf{w}_{\mathrm{a}}\|\nonumber\\
    &\le |\mathbf{v}_{\mathrm{w}_j}^{\dagger} \overline{\mathbf{H}}_{\mathrm{a},\mathrm{w}_{j}}  \mathbf{w}_{\mathrm{a}}|+ \delta_{\mathrm{w}_j}. \label{eq:DEP_inequality}
\end{align}    
\end{footnotesize}Thus, \eqref{cons:imperfect_DEP_1} can be reformulated into the following form

\vspace{-2mm}
\begin{footnotesize}
    \begin{align}
     & \;{ P_\mathrm{a}(1 + K(|\mathbf{v}_{\mathrm{w}_j}^{\dagger} \overline{\mathbf{H}}_{\mathrm{a},\mathrm{w}_{j}}  \mathbf{w}_{\mathrm{a}}|+ \delta_{\mathrm{w}_j})^2})  \leq \eta_{\mathrm{w}_j},\nonumber\\&\quad \forall j \in \{1, 2, \dots, n_\mathrm{w}\}. \label{cons:DEP_imperfect_rewritten}
\end{align}
\end{footnotesize}Note that \eqref{cons:DEP_imperfect_rewritten} is still intractable. By applying the inequality $(x + y)^2 \leq 2(x^2 + y^2)$, we can obtain an upper bound on the left-hand side (LHS) of~\eqref{cons:DEP_imperfect_rewritten}, which is expressed~as

\vspace{-3mm}
\begin{footnotesize}
 \begin{align}
    &P_\mathrm{a}(1 + K(|\mathbf{v}_{\mathrm{w}_j}^{\dagger} \overline{\mathbf{H}}_{\mathrm{a},\mathrm{w}_{j}}  \mathbf{w}_{\mathrm{a}}|+ \delta_{\mathrm{w}_j})^2)\nonumber\\
    &\leq P_\mathrm{a} + 2KP_\mathrm{a}|\mathbf{v}_{\mathrm{w}_j}^{\dagger} \overline{\mathbf{H}}_{\mathrm{a},\mathrm{w}_{j}}  \mathbf{w}_{\mathrm{a}}|^2 + 2KP_\mathrm{a}\delta_{\mathrm{w}_j}^2.
\end{align}   
\end{footnotesize}Then, the covertness requirement under imperfect channel estimation can be given by

\vspace{-2mm}
\begin{footnotesize}
\begin{align}
     &P_\mathrm{a} + 2KP_\mathrm{a}|\mathbf{v}_{\mathrm{w}_j}^{\dagger} \overline{\mathbf{H}}_{\mathrm{a},\mathrm{w}_{j}}  \mathbf{w}_{\mathrm{a}}|^2 + 2KP_\mathrm{a}\delta_{\mathrm{w}_j}^2 \leq \eta_{\mathrm{w}_j}, \nonumber \\
     & \quad \forall j \in \{1, 2, \dots, n_\mathrm{w}\}.\label{cons:imperfect_DEP}
\end{align}      
\end{footnotesize}

\vspace{-5pt}
\subsection{\hspace{-5pt}Reliability Requirement \hspace{-0.5pt}under \hspace{-0.5pt}Imperfect \hspace{-0.5pt}Channel \hspace{-0.5pt}Estimation}
Regarding the reliability requirement under imperfect channel estimation, we also consider that the reliability requirement~\eqref{cons:Pout_req_simplify}, which is derived under perfect channel estimation, can be satisfied under all possible channel errors. Thus, the reliability requirement under imperfect channel estimation is given by
\begin{equation}
\footnotesize
    \min_{\Delta \mathbf{H}_{\mathrm{a},\mathrm{b}} \in \Omega_{\mathrm{a},\mathrm{b}}}P_\mathrm{a}\eta_\mathrm{b}|\mathbf{v}_{\mathrm{b}}^{\dagger}\mathbf{H}_{\mathrm{a},\mathrm{b}} \mathbf{w}_\mathrm{a}|^2 \geq r_{\mathrm{a}}.\label{cons:imperfect_TOP_1}
\end{equation}
Similar to \eqref{eq:DEP_inequality}, as for~\eqref{cons:imperfect_TOP_1}, we have

\vspace{-2mm}
\begin{footnotesize}
    \begin{align}
        &|\mathbf{v}_{\mathrm{b}}^{\dagger}({\mathbf{H}}_{\mathrm{a}, \mathrm{b}} + \Delta \mathbf{H}_{\mathrm{a}, \mathrm{b}}) \mathbf{w}_\mathrm{a}|\nonumber\\&\geq |\mathbf{v}_{\mathrm{b}}^{\dagger}{\mathbf{H}}_{\mathrm{a}, \mathrm{b}}\mathbf{w}_\mathrm{a}|-|\mathbf{v}_{\mathrm{b}}^{\dagger} \Delta \mathbf{H}_{\mathrm{a}, \mathrm{b}} \mathbf{w}_\mathrm{a}|\nonumber\\
        &\geq |\mathbf{v}_{\mathrm{b}}^{\dagger}{\mathbf{H}}_{\mathrm{a}, \mathrm{b}}\mathbf{w}_\mathrm{a}|-\|\mathbf{v}_{\mathrm{b}}\| \cdot \|\Delta \mathbf{H}_{\mathrm{a}, \mathrm{b}} \| \cdot\|\mathbf{w}_\mathrm{a}\|\nonumber\\
        &\geq |\mathbf{v}_{\mathrm{b}}^{\dagger}{\mathbf{H}}_{\mathrm{a}, \mathrm{b}}\mathbf{w}_\mathrm{a}|-\delta_{\mathrm{b}}. \label{cons:TOP_imperfect_rewritten}
    \end{align}
\end{footnotesize}Thus, \eqref{cons:imperfect_TOP_1} can be reformulated as

\vspace{-2mm}
\begin{footnotesize}
\begin{align}
    &P_\mathrm{a}\eta_\mathrm{b}(|\mathbf{v}_{\mathrm{b}}^{\dagger}{\mathbf{H}}_{\mathrm{a}, \mathrm{b}}\mathbf{w}_\mathrm{a}|^2-2\delta_{\mathrm{b}}|\mathbf{v}_{\mathrm{b}}^{\dagger}{\mathbf{H}}_{\mathrm{a}, \mathrm{b}}\mathbf{w}_\mathrm{a}|+\delta_{\mathrm{b}}^2)   \geq {r_{\mathrm{a}}}.\label{cons:TOP_imperfect}
\end{align}    
\end{footnotesize}Note that \eqref{cons:TOP_imperfect} remains non-convex. To address this, we provide a lower bound on the LHS of \eqref{cons:TOP_imperfect}, 
which can be expressed as

\vspace{-2mm}
\begin{footnotesize}
\begin{align}
    &P_\mathrm{a}\eta_\mathrm{b}(|\mathbf{v}_{\mathrm{b}}^{\dagger}{\mathbf{H}}_{\mathrm{a}, \mathrm{b}}\mathbf{w}_\mathrm{a}|^2-2\delta_{\mathrm{b}}|\mathbf{v}_{\mathrm{b}}^{\dagger}{\mathbf{H}}_{\mathrm{a}, \mathrm{b}}\mathbf{w}_\mathrm{a}|+\delta_{\mathrm{b}}^2)\nonumber\\ &   \geq P_\mathrm{a}\eta_\mathrm{b}(|\mathbf{v}_{\mathrm{b}}^{\dagger}{\mathbf{H}}_{\mathrm{a}, \mathrm{b}}\mathbf{w}_\mathrm{a}|^2-2\delta_{\mathrm{b}}\max\left(|\mathbf{v}_{\mathrm{b}}^{\dagger}{\mathbf{H}}_{\mathrm{a}, \mathrm{b}}\mathbf{w}_\mathrm{a}|^2,1\right)+\delta_{\mathrm{b}}^2). \label{cons:TOP_imperfect_relax}
\end{align}    
\end{footnotesize}Then, the reliability requirement under imperfect channel estimation can be expressed as 
\begin{equation}
\footnotesize
     P_\mathrm{a}\eta_\mathrm{b}(|\mathbf{v}_{\mathrm{b}}^{\dagger}{\mathbf{H}}_{\mathrm{a}, \mathrm{b}}\mathbf{w}_\mathrm{a}|^2-2\delta_{\mathrm{b}}\max\left(|\mathbf{v}_{\mathrm{b}}^{\dagger}{\mathbf{H}}_{\mathrm{a}, \mathrm{b}}\mathbf{w}_\mathrm{a}|^2,1\right)+\delta_{\mathrm{b}}^2) \geq {r_{\mathrm{a}}}. \label{cons:imperfect_TOP}
\end{equation}

\vspace{-3pt}
\subsection{Beamforming Design under Imperfect Channel Estimation}

We now investigate the optimal beamforming design for CR maximization under imperfect channel estimation. Following \eqref{Prob:Beamforming_perfect_SDR}, \eqref{cons:imperfect_DEP} and \eqref{cons:imperfect_TOP}, such an optimization problem can be 

\hspace{-2.9mm}formulated as

\vspace{-4mm}
\begin{subequations}
\begin{footnotesize}
    \label{Prob:Opt_Beamforming_imperfect}
    \begin{align}
        & \max _{ \mathbf{w}_{\mathrm{a}}, P_\mathrm{a}, r_{\mathrm{a}}}\; r_{\mathrm{a}} \\
       \text { s.t. }
        &\;\text{\eqref{wa_normal_cons}, \eqref{alice_power_cons}, \eqref{cons:imperfect_DEP}, \eqref{cons:imperfect_TOP}.}
    \end{align}
    \end{footnotesize}
\end{subequations}

{Same as problem \eqref{Prob:Beamforming_perfect_simplify}}, problem~\eqref{Prob:Opt_Beamforming_imperfect} can be reformulated~as

\vspace{-3mm}
\begin{footnotesize}
\begin{subequations}
    \label{Prob:Beamforming_imperfect_SDR_LMI}
    \begin{align}
        & \max _{ \widehat{\mathbf{W}}_{\mathrm{a}}, r_{\mathrm{a}}}\; r_{\mathrm{a}} \\
       \text { s.t. }
        &\; (1+2K\delta_{\mathrm{w}_j}^2)\mathrm{Tr}(\widehat{\mathbf{W}}_{\mathrm{a}})+2K\,\mathrm{Tr}(\mathbf{V}_{\mathrm{w}_j}\overline{\mathbf{H}}_{\mathrm{a},\mathrm{w}_{j}}\widehat{\mathbf{W}}_{\mathrm{a}}\overline{\mathbf{H}}_{\mathrm{a},\mathrm{w}_{j}}^{\dagger})\leq \eta_{\mathrm{w}_j},\nonumber\\
        & \; \forall j \in \{1, 2, \dots, n_\mathrm{w}\},\label{cons:DEP_imperfect_1}\\
        &\; \mathrm{Tr}(\mathbf{V}_{\mathrm{b}}{\mathbf{H}}_{\mathrm{a},\mathrm{b}}\widehat{\mathbf{W}}_{\mathrm{a}}{\mathbf{H}}_{\mathrm{a},\mathrm{b}}^{\dagger})+\delta_{\mathrm{b}}^2\mathrm{Tr}(\widehat{\mathbf{W}}_{\mathrm{a}})\nonumber\\
        &-2\delta_{\mathrm{b}}\max\left(\mathrm{Tr}(\mathbf{V}_{\mathrm{b}}{\mathbf{H}}_{\mathrm{a},\mathrm{b}}\widehat{\mathbf{W}}_{\mathrm{a}}{\mathbf{H}}_{\mathrm{a},\mathrm{b}}^{\dagger}),\mathrm{Tr}(\widehat{\mathbf{W}}_{\mathrm{a}})\right) \geq {r_{\mathrm{a}}}/{\eta_\mathrm{b}},\label{cons:TOP_imperfect_1}\\
        &\;\text{\eqref{wa_power_cons_SDR}, \eqref{SD_cons}, \eqref{SD_cons2}}.
    \end{align}
\end{subequations}
\end{footnotesize}Due to the constraint \eqref{SD_cons2} being non-convex, problem~\eqref{Prob:Beamforming_imperfect_SDR_LMI} is difficult to solve. Similar to problem~\eqref{Prob:Beamforming_perfect_SDR}, we can first utilize the SDR technique to relax \eqref{SD_cons2} and reformulate problem~\eqref{Prob:Beamforming_imperfect_SDR_LMI} into an SDP problem. Then, the standard convex optimization solver is used to solve the SDP problem, and the Gaussian randomization procedure is used to generate high-quality rank-one solutions. The optimal solutions and the corresponding CR $R_{\mathrm{a}}^*$ can also be obtained using Algorithm~\ref{algorithm:Opt-JBJ}.

\vspace{-8pt}
\subsection{Joint Beamforming and Antenna Orientation Design under Imperfect Channel Estimation }
\vspace{-3pt}

In this subsection, we explore the joint beamforming and antenna orientation design to maximize CR under imperfect channel estimation.
Following \eqref{Prob:Beamforming_antenna_perfect}, \eqref{cons:imperfect_DEP} and \eqref{cons:imperfect_TOP}, the CR maximization problem can be formulated as

\vspace{-4.5mm}
\begin{footnotesize}
\begin{subequations}
    \label{Prob:Beamforming_antenna_imperfect}
    \begin{align}
        & \max _{\mathbf{o}_{\mathrm{a}}, \mathbf{w}_{\mathrm{a}}, P_\mathrm{a}, r_{\mathrm{a}}}\; r_{\mathrm{a}} \\
       \text { s.t. }
        &\;\text{\eqref{wa_normal_cons}, \eqref{alice_power_cons}, \eqref{cons:imperfect_DEP}, \eqref{cons:imperfect_TOP}.}
    \end{align}
\end{subequations}
\end{footnotesize}Similar to problem \eqref{Prob:antenna_perfect}, we also convert \eqref{Prob:Beamforming_antenna_imperfect} into two optimization subproblems.
The first one concerns the optimization of $\mathbf{W}_{\mathrm{a}}$ and $r_{\mathrm{a}}$ under the given $P_\mathrm{a}$ and $\mathbf{o}_{\mathrm{a}}$, while the second one deals with the optimization of the $\mathbf{o}_{\mathrm{a}}$ and $P_\mathrm{a}$ under the given optimal beamforming vector $\mathbf{W}_{\mathrm{a}}$.
The first subproblem can be formulated as 

\vspace{-4.5mm}
\begin{footnotesize}
\begin{subequations}
    \label{Prob:O-BA_sub_problem_1_imp}
    \begin{align}
        & \max _{ {\mathbf{W}}_{\mathrm{a}}, r_{\mathrm{a}}}\; r_{\mathrm{a}} \\
       \text { s.t. }
        & P_\mathrm{a}(1 + 2K(\mathrm{Tr}(\mathbf{V}_{\mathrm{w}_j}\overline{\mathbf{H}}_{\mathrm{a},\mathrm{w}_{j}}\mathbf{W}_\mathrm{a}\overline{\mathbf{H}}_{\mathrm{a},\mathrm{w}_{j}}^{\dagger}) + \delta_{\mathrm{w}_j}^2))\leq \eta_{\mathrm{w}_j},\nonumber\\
       &\; \forall j \in \{1, 2, \dots, n_\mathrm{w}\}, \label{cons:imperfect_DEP_simp}\\
       &\; \mathrm{Tr}(\mathbf{V}_{\mathrm{b}}{\mathbf{H}}_{\mathrm{a},\mathrm{b}}{\mathbf{W}}_{\mathrm{a}}{\mathbf{H}}_{\mathrm{a},\mathrm{b}}^{\dagger})-2\delta_{\mathrm{b}}\max\left(\mathrm{Tr}(\mathbf{V}_{\mathrm{b}}{\mathbf{H}}_{\mathrm{a},\mathrm{b}}{\mathbf{W}}_{\mathrm{a}}{\mathbf{H}}_{\mathrm{a},\mathrm{b}}^{\dagger}),1\right)+\delta_{\mathrm{b}}^2\nonumber\\
        & \geq {r_{\mathrm{a}}}/{(P_\mathrm{a}\eta_\mathrm{b})},\label{cons:imperfect_TOP_simp}\\
        &\;\text{\eqref{cons:Tr_Wa_1}, \eqref{SD_cons_Wa}}.
        \vspace{-2mm}
    \end{align}
\end{subequations}
\end{footnotesize}By using a standard convex optimization solver, the optimal beamforming matrix ${\mathbf{W}}_{\mathrm{a}}^*$ of the problem \eqref{Prob:O-BA_sub_problem_1_imp} can be obtained.

With the optimal $\mathbf{W}_\mathrm{a}^*$, the second subproblem can be formulated as

\vspace{-4.5mm}
\begin{footnotesize}
\begin{subequations}
    \label{Prob:antenna_imperfect_Gab}
    \begin{align}
        & \max _{ \mathbf{o}_{\mathrm{a}}, P_\mathrm{a}}\; G_{\mathrm{a},\mathrm{b}} P_\mathrm{a}  \\
       \text { s.t. }
        &\text{ \eqref{alice_power_cons}, \eqref{cons:imperfect_DEP_simp}, \eqref{cons:imperfect_TOP_simp}.}
    \end{align}
\end{subequations}
\end{footnotesize}With Theorem~\ref{Theorem:minimum_psi_ab}, problem~\eqref{Prob:antenna_imperfect_Gab} can also be simplified.
Therefore, by replacing subproblems~\eqref{Prob:antenna_perfect_wa} and~\eqref{Prob:antenna_perfect_Gab} with~\eqref{Prob:O-BA_sub_problem_1_imp} and~\eqref{Prob:antenna_imperfect_Gab} in Algorithm~\ref{overall_algorithm_P1}, the same algorithm can be applied to solve problem~\eqref{Prob:Beamforming_antenna_imperfect}, just as it is used for solving problem~\eqref{Prob:Beamforming_antenna_perfect}.

\vspace{-7pt}
\section{Simulation and Numerical Results}
\vspace{-3pt}
\label{sec:Simu}
In this section, we provide extensive simulation and numerical results to illustrate the performance of the optimal beamforming (OB) design and the joint optimal beamforming and antenna orientation (JO-BA) design under both perfect and imperfect channel estimations.
{Moreover, we compare the CR performance of the proposed OB/JO-BA designs with several benchmark covert beamforming schemes, including zero-forcing (ZF)~\cite{Ma2021} and maximum-ratio transmission (MRT)~\cite{Forouzesh2020a}.}
Unless otherwise specified, the setting of system parameters is summarized in Table~\ref{table1}.

\vspace{-9pt}
\begin{table}[htbp]
	\centering  
	\caption{System Parameters}  
        \vspace{-4pt}
	\label{table1}  
        \renewcommand{\arraystretch}{1.2}
	\begin{tabular}{|@{\hskip 3pt}p{0.258\textwidth}@{\hskip 3pt}|@{\hskip 3pt}p{0.2\textwidth}@{\hskip 3pt}|}  
		\hline  
		\textbf{Parameters} & \textbf{Values} \\  
		\hline
            Carrier frequency & $18$ \text{GHz} (Ku) \\
		\hline
            Near-in side-lobe level $L_S$ & $-6.75$ \text{dBi} \\
		\hline
            Far-out side-lobe level $L_F$ & $0$ \text{dBi}\\
            \hline
            $3$dB beamwidth $\psi_{3\mathrm{dB}}$ & $0.4^{\circ}$ \\
            \hline
            Minimum off-boresight angel $\theta_0$ & $1^{\circ}$ \\
            \hline
            Maximum antenna gain of Earth station and satellites $G_{\mathrm{a},\mathrm{max}}^{\text{dBi}}$, $G_{\mathrm{sat},\mathrm{max}}^{\text{dBi}}$ & $32$ \text{dBi} \\
            \hline
            Maximum transmit power $P_\mathrm{a,max}$ & $50$ \text{dBm} \\
            \hline
            Number of antennas at Alice $M_{\mathrm{a}}$ & $8\times8=64$ \\
		\hline
            Number of antennas at satellites $M_{\mathrm{sat}}$ & $4\times4=16$ \\
		\hline
            Noise power at Bob and wardens $\sigma_{\mathrm{b}}^2,\sigma_\mathrm{w}^2$ & $-90$ \text{dBm}  \\
            \hline
            Noise uncertainty at Bob and wardens $\rho$ & $1.5$\\
            \hline
            Rician factor $K$ & $7$ \\
            \hline
            Radius of Earth $d_{\mathrm{a}}$ &$6,378$ km\\
            \hline
            Satellite altitude $h$ & $36,000$ km \\
            \hline
            Location of Alice $\mathbf{q}_\mathrm{a}$ & $\left[0,\,d_{\mathrm{a}},\,0 \right]^T$ \\
            \hline
            Location of Bob $\mathbf{q}_\mathrm{b}$ & $d_{\mathrm{sat}}[ \cos90^{\circ},\,\sin90^{\circ},\, 0]^T$  \\
            \hline
            Location of wardens $\mathbf{q}_{\mathrm{w}_1}$,$\mathbf{q}_{\mathrm{w}_2}$,$\mathbf{q}_{\mathrm{w}_3}$,$\mathbf{q}_{\mathrm{w}_4}$ &$d_{\mathrm{sat}}[\cos92^{\circ}, \sin92^{\circ}, 0]^T$,\\
            &$d_{\mathrm{sat}}[\cos91^{\circ}, \sin91^{\circ}, 0]^T$,\\
            &$d_{\mathrm{sat}}[\cos89^{\circ}, \sin89^{\circ}, 0]^T$,\\
            &$d_{\mathrm{sat}}[\cos88^{\circ}, \sin88^{\circ}, 0]^T$\\
            \hline
	\end{tabular}
\vspace{-7pt}
\end{table}

\begin{figure}[bh]
    \vspace{-10pt}
    \centering
    \includegraphics[width=2.1in]{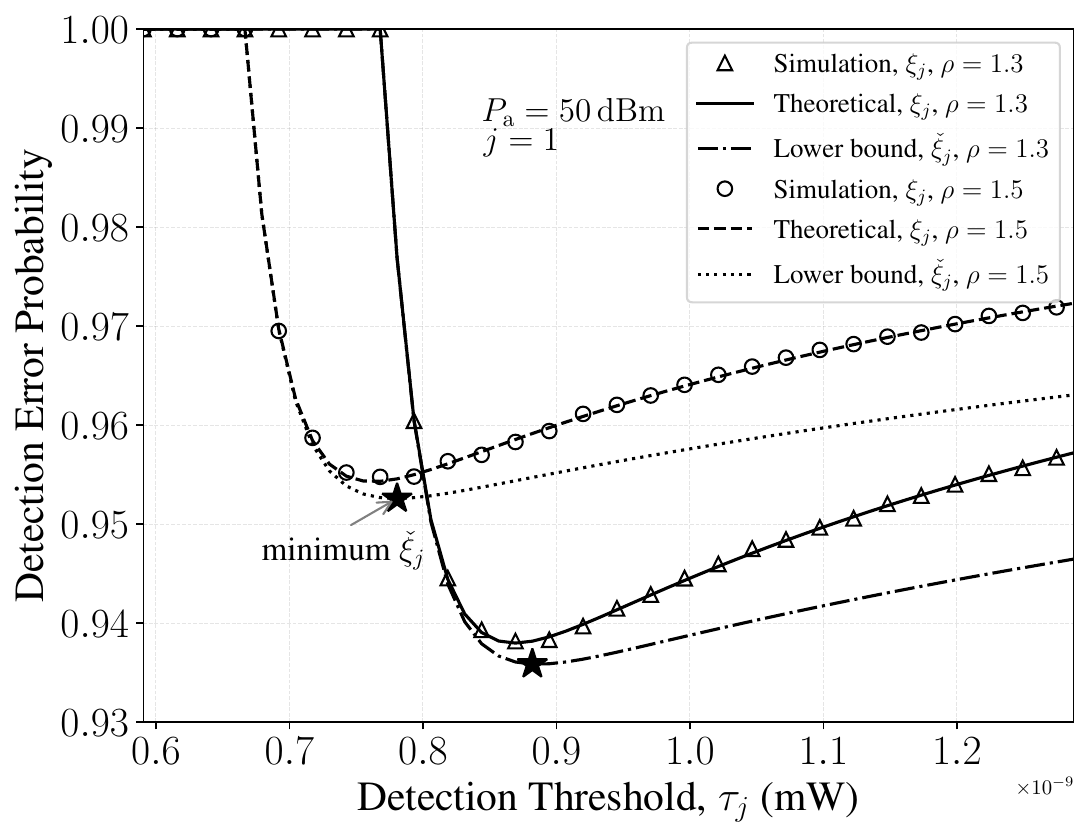}
    \vspace{-12pt}
    \caption{{Detection error probability vs. detection threshold $\tau_j$.}}
    \vspace{-13pt}
    \label{fig:DEP_vs_threshold}
\end{figure}

\vspace{-8pt}
\subsection{Model Validation}
\vspace{-3pt}
To explore the impact of the detection threshold $\tau_j$ on the DEP $\xi_{j}$ as well as its lower bound $\check{\xi}_{j}$, we conduct extensive simulations over $10^6$ independent channel realizations and summarize in Fig.~\ref{fig:DEP_vs_threshold} how $\xi_{j}$ and $\check{\xi}_{j}$ vary with $\tau_j$ for the setting of $j=1$, $P_{\mathrm{a}}=50\,\mathrm{dBm}$ and $\rho =\{1.2,\,1.5\}$. 
It can be seen from Fig.~\ref{fig:DEP_vs_threshold} that the theoretical results of $\xi_{j}$ well match the simulation ones, which validates the correctness of our theoretical modeling for $\xi_{j}$. 
Another observation from Fig.~\ref{fig:DEP_vs_threshold} is that $\check{\xi}_{j}$ is consistently lower than $\xi_{j}$ for a given $\rho$, and that there exists an optimal detection threshold $\tau_j^{*}$ minimizing $\check{\xi}_{j}$, which is consistent with Lemma~\ref{Theorem:F_z_lb} and Theorem~\ref{Theorem:Opt_tau}.

{
\begin{figure}[ht]
        \vspace{-7pt}
	\centering
        \hspace{-10pt}
	\subfigure[Minimum DEP vs. $P_\mathrm{a}$]{
            \vspace{-15pt}
		\includegraphics[width=0.495\linewidth]{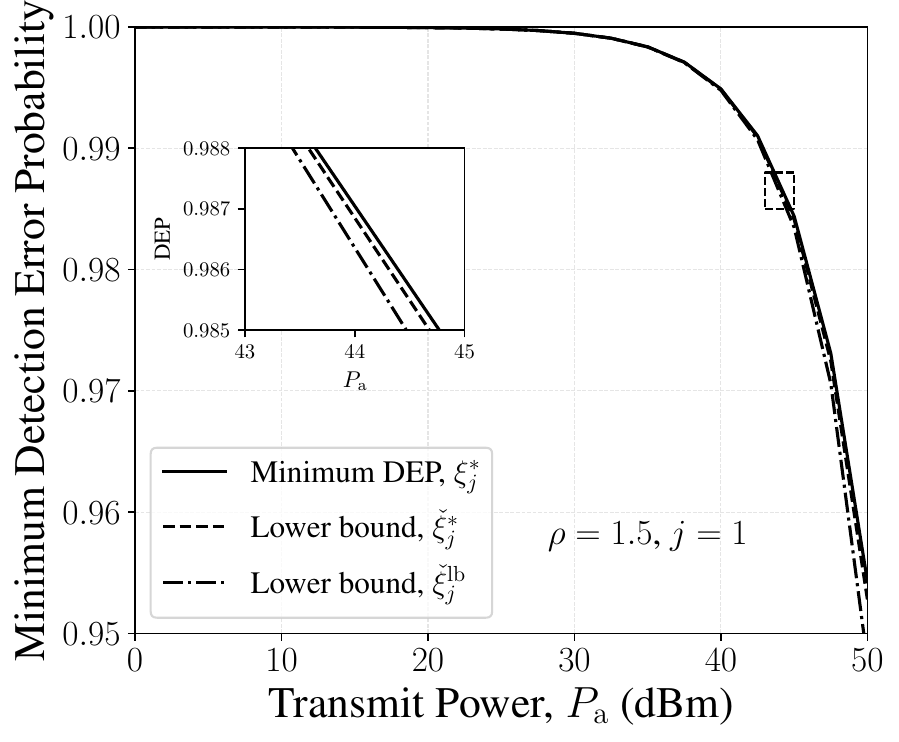}
		\label{fig:DEP_vs_pa}
	}
        \hspace{-15pt}
	\subfigure[Minimum DEP vs. $\rho$]{
            \vspace{-15pt}
		\includegraphics[width=0.495\linewidth]{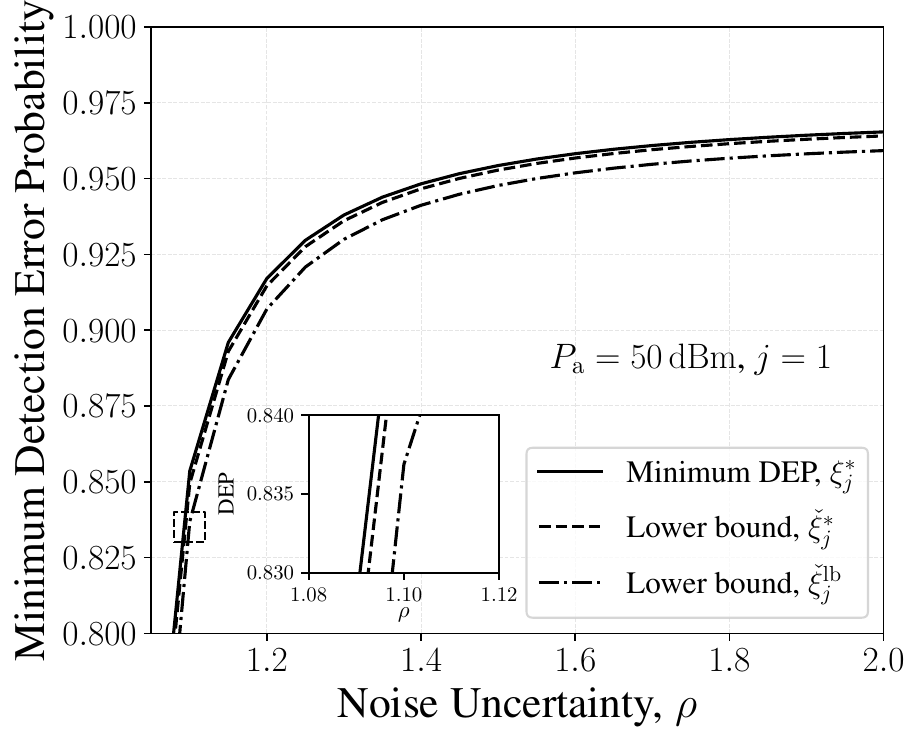}
		\label{fig:DEP_vs_rho}
	}
        \hspace{-15pt}
        \vspace{-7pt}
	\caption{{Validation of minimum DEP under different parameter settings.}}
        \label{fig:DEP_vs_pa_rho}
        \vspace{-15pt}
\end{figure}
}

To validate the effectiveness of adopting the lower bound of DEP as the covertness requirement, we compare the approximations $\check{\xi}_{j}^*$, $\check{\xi}^{\mathrm{lb}}_{j}$ with the minimum value of DEP $\xi_j$, denoted by $\xi_j^*$, which is obtained via the exhaustive search.
Under the setting of $P_{\mathrm{a}}=50\, \mathrm{dBm}$ and $\rho=1.5$,
the impacts of $P_{\mathrm{a}}$ and $\rho$ on $\xi_{j}^*$,  $\check{\xi}_{j}^*$, and $\check{\xi}^{\mathrm{lb}}_{j}$ are summarized in Fig.~\ref{fig:DEP_vs_pa} and Fig.~\ref{fig:DEP_vs_rho}, respectively. 
We can see from Fig.~\ref{fig:DEP_vs_pa_rho} that both $\check{\xi}^{*}_{j}$ and $\check{\xi}^{\mathrm{lb}}_{j}$ closely approximate $\xi_{j}^{*}$. Furthermore, the discrepancy between $\xi_{j}^*$ and its approximations $\check{\xi}_{j}$, $\check{\xi}^{\mathrm{lb}}_{j}$ decreases monotonically with decreasing $P_{\mathrm{a}}$ and $\rho$.
These phenomena indicate the tightness of our covertness requirement and are consistent with Lemma~\ref{Theorem:F_z_lb} and Theorem~\ref{Theorem:DEP_lower_bound}.

\begin{figure*}[hb]
\vspace{-12pt}
    \centering
    \begin{minipage}[t]{0.3\textwidth}
    \centering
    \raisebox{0pt}{\includegraphics[width=\textwidth]{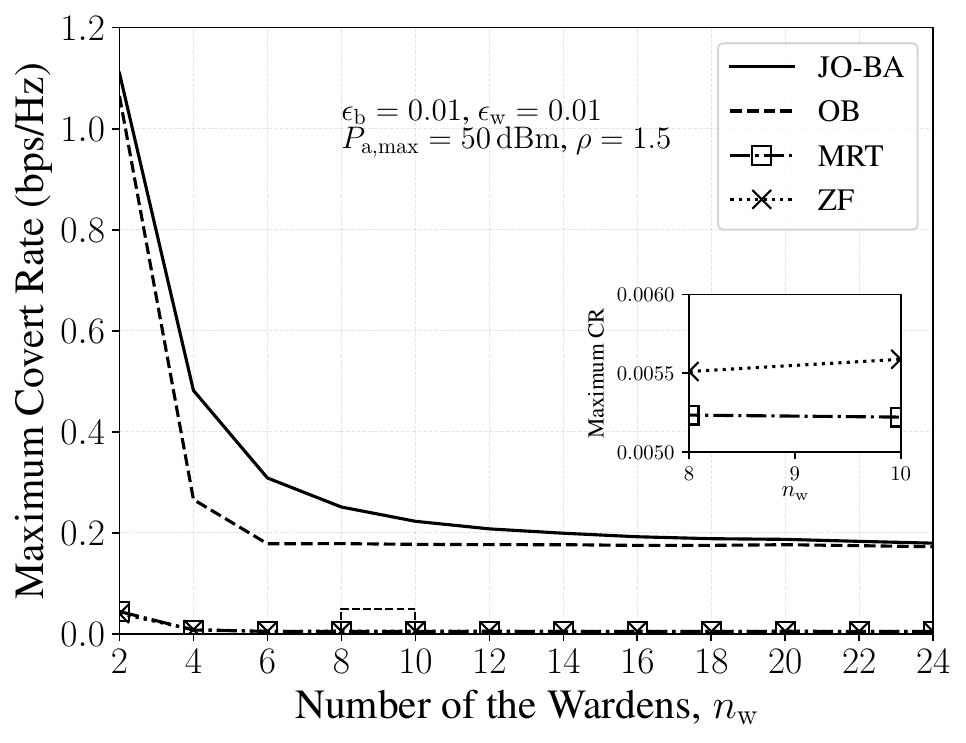}}
    \vspace{-23pt}
    \caption{{Maximum CR vs. number of the wardens $n_{\mathrm{w}}$.}}
    \vspace{-12pt}
    \label{fig:CR_vs_nw}
    \end{minipage}
    \hfill
    \begin{minipage}[t]{0.3\textwidth}
        \centering
        \raisebox{3pt}{\includegraphics[width=\textwidth]{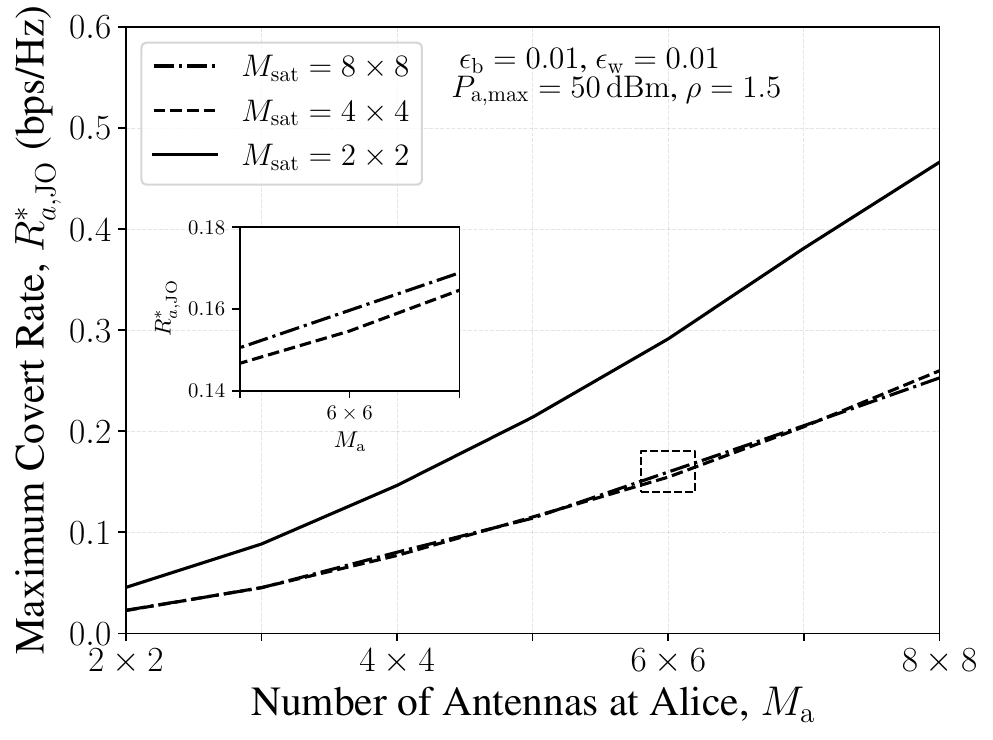}}
        \vspace{-25pt}
        \caption{{Maximum CR achieved by JO-BA vs. number of antennas at Alice $M_\mathrm{a}$.}}
        \vspace{-12pt}
        \label{fig:CR_vs_antenna}
    \end{minipage}
    \hfill
    \begin{minipage}[t]{0.3\textwidth}
        \centering
        \raisebox{2pt}{\includegraphics[width=\textwidth]{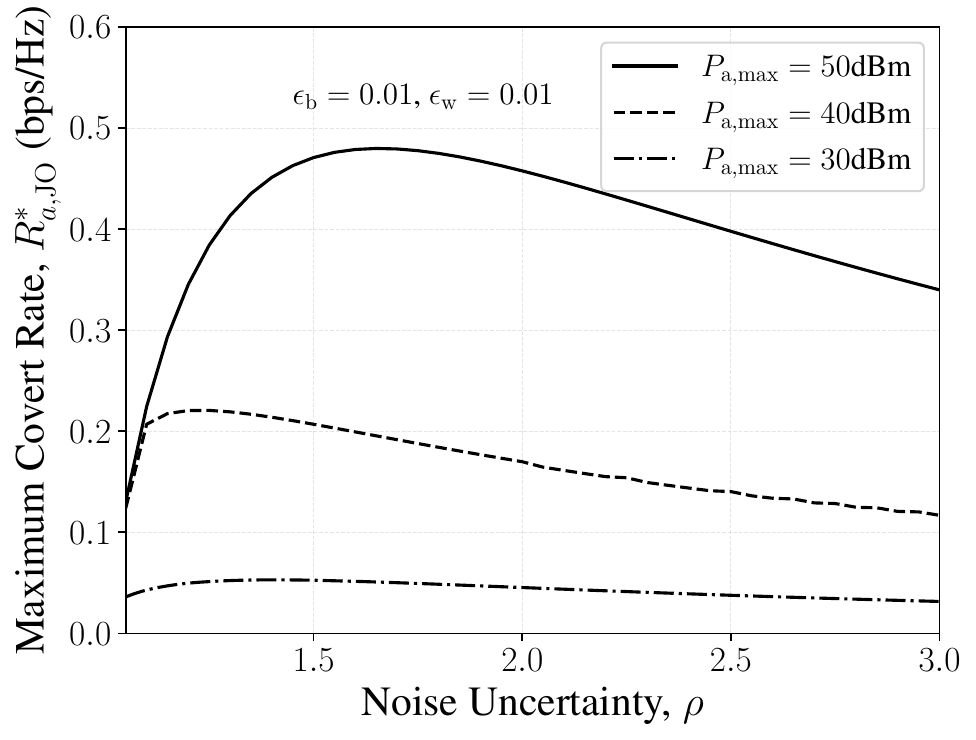}}
        \vspace{-25pt}
        \caption{{Maximum CR achieved by JO-BA vs. noise uncertainty $\rho$.}}
        \vspace{-12pt}
        \label{fig:CR_vs_rho}
    \end{minipage}
\end{figure*}

\vspace{-10pt}
\subsection{CR Performance under Perfect Channel Estimation}
\vspace{-2pt}

We now investigate the impact of the number of wardens $n_\mathrm{w}$ on the maximum CR under perfect channel estimation. Fig.~\ref{fig:CR_vs_nw} presents how the maximum CR varies with $n_\mathrm{w}$ for the setting of $P_\mathrm{a,max}=50$~dBm, $\rho=1.5$, and $\epsilon_{\mathrm{b}} = \epsilon_\mathrm{w} = 0.01$.
In this simulation, the $n_\mathrm{w}$ wardens are co-orbital with Bob and uniformly distributed around Bob within Alice's observable range, where the azimuth angles $\psi_{\mathrm{w}_j}$ of the wardens follow a uniform distribution over \([88^\circ, 92^\circ]\).
{Fig.~\ref{fig:CR_vs_nw} also shows that the proposed OB and JO-BA schemes outperform the MRT and ZF schemes in terms of CR performance. 
The reasons can be explained as follows. 
Since the angular separation between Bob and the wardens is small, the Alice-Bob link is similar to the Alice-warden links. When the MRT scheme maximizes the beamforming gain toward Bob, it simultaneously enhances beamforming gains toward wardens, which forces Alice reduce transmit power, leading to a low CR performance. Similarly, when the ZF scheme forces the beamforming gain toward all wardens to zero, the gain toward Bob also approaches zero, resulting in a low CR performance.}
We can also observe from Fig.~\ref{fig:CR_vs_nw} that when $n_{\mathrm{w}}$ is either sufficiently small or large, the difference of the maximum CR achieved by OB and JO-BA is not significant. 
This is because the distance between Bob and the wardens is either too large or too small when $n_{\mathrm{w}}$ is either sufficiently small or large, making the impact of antenna orientation optimization negligible.

The impact of the antenna array sizes at Alice and the satellites on the maximum CR obtained by JO-BA is further investigated under the perfect channel estimation.
Under each setting of $M_{\mathrm{sat}}=\{2\times2,\,4\times4,\,8\times8\}$, $P_\mathrm{a,max}=50$~dBm, $\rho=1.5$, and $\epsilon_{\mathrm{b}} =\epsilon_\mathrm{w} = 0.01$, Fig. \ref{fig:CR_vs_antenna} illustrates how the maximum CR denoted as $R_{\mathrm{a},\mathrm{JO}}^*$ varies with $M_{\mathrm{a}}$ for different numbers of antennas at the satellites $M_{\mathrm{sat}}$.
We can observe from Fig.~\ref{fig:CR_vs_antenna} that $R_{\mathrm{a},\mathrm{JO}}^*$ increases with $M_{\mathrm{a}}$.
{This is because increasing $M_{\mathrm{a}}$ enables more
precise beam pattern steering toward Bob to achieve higher beamforming gain, and simultaneously results in less signal leakage to the wardens. We can also see from Fig.~\ref{fig:CR_vs_antenna} that a higher $M_{\mathrm{sat}}$ results in a lower $R_{\mathrm{a},\mathrm{JO}}^*$. 
This is because increasing $M_{\mathrm{sat}}$ enhances the receive beamforming gain for both Bob and the wardens, which improves the SNR at Bob while simultaneously reducing the DEP at the wardens. However, the reduced DEP forces Alice to lower its transmit power. Thus, the covert rate decreases, and this power reduction effect ultimately dominates the covert rate performance.}

To explore the effects of the noise uncertainty $\rho$ on the maximum CR $R_{\mathrm{a},\mathrm{JO}}^*$ under perfect channel estimation, we summarize in Fig.~\ref{fig:CR_vs_rho} how $R_{\mathrm{a},\mathrm{JO}}^*$ varies with $\rho$ for each setting of $P_{\mathrm{a},\max} = \{30,40,50\}$~dBm, and $\epsilon_{\mathrm{b}}=\epsilon_\mathrm{w} = 0.01$.
We can see from Fig.~\ref{fig:CR_vs_rho} that as $\rho$ increases, $R_{\mathrm{a},\mathrm{JO}}^*$ first increases and then decreases.
This can be explained as follows. 
Note that the DEP of the $j^{\mathrm{th}}$ warden and the TOP of the Alice–Bob link both increase as $\rho$ increases.
When $\rho$ is relatively small, the increase of DEP dominates that of TOP, which allows Alice to increase transmit power $P_{\mathrm{a}}$ up to some threshold under the constraint of the covertness requirement, thus $R_{\mathrm{a},\mathrm{JO}}^*$ can increase to the maximum value. 
As $\rho$ continues to increase, the increase of TOP dominates that of DEP, requiring a reduction in the transmit rate to meet the constraint of the reliability requirement, thereby leading to a decrease in $R_{\mathrm{a},\mathrm{JO}}^*$.

\begin{figure*}[thb]
    \centering
    \begin{minipage}[t]{0.3\textwidth}
        \centering
        \raisebox{3pt}{\includegraphics[width=\textwidth]{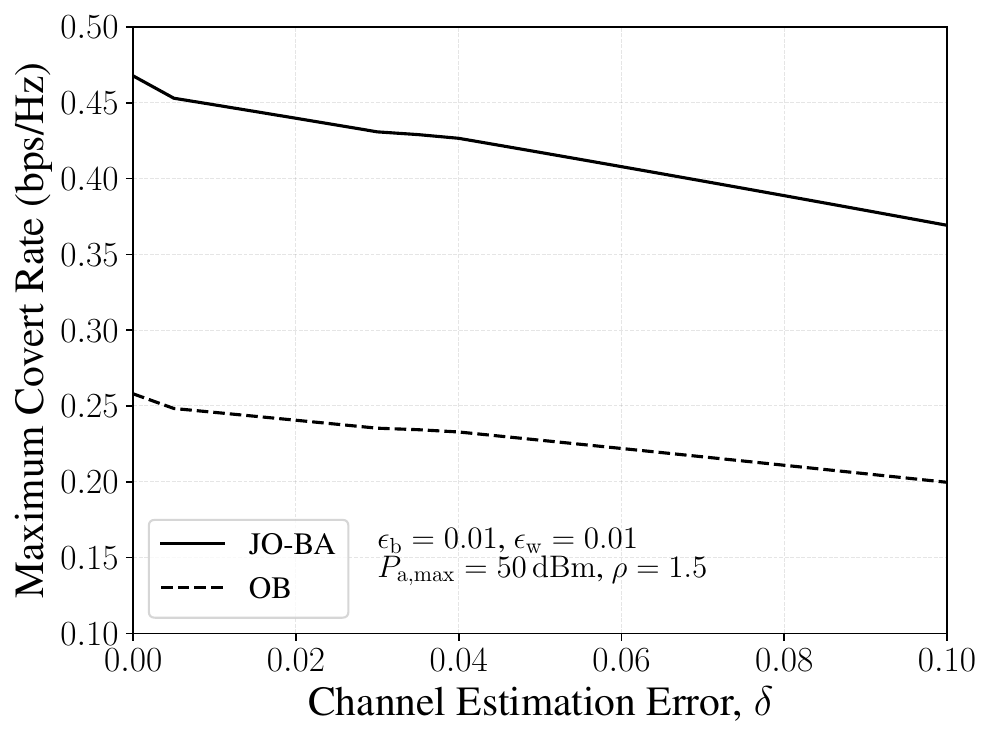}}
        \vspace{-23pt}
        \caption{{Maximum CR vs. channel estimation error $\delta$.}}
        \vspace{-15pt}
        \label{fig:CR_vs_delta}
    \end{minipage}
    \hfill
    \begin{minipage}[t]{0.3\textwidth}
        \centering
        \raisebox{3pt}{\includegraphics[width=\textwidth]{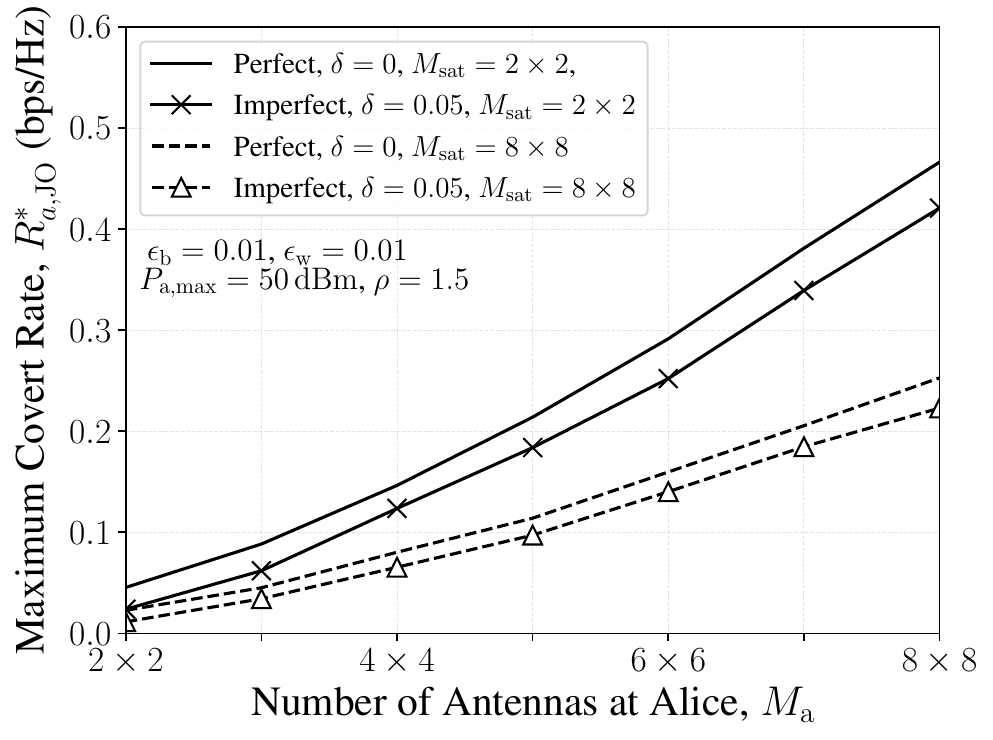}}
        \vspace{-25pt}
        \caption{{Maximum CR achieved by JO-BA vs. $M_{\mathrm{a}}$ under perfect/imperfect channel estimation.}}
        \vspace{-15pt}
        \label{fig:CR_vs_antenna_imp}
    \end{minipage}
    \hfill
    \begin{minipage}[t]{0.3\textwidth}
        \centering
        \raisebox{2pt}{\includegraphics[width=\textwidth]{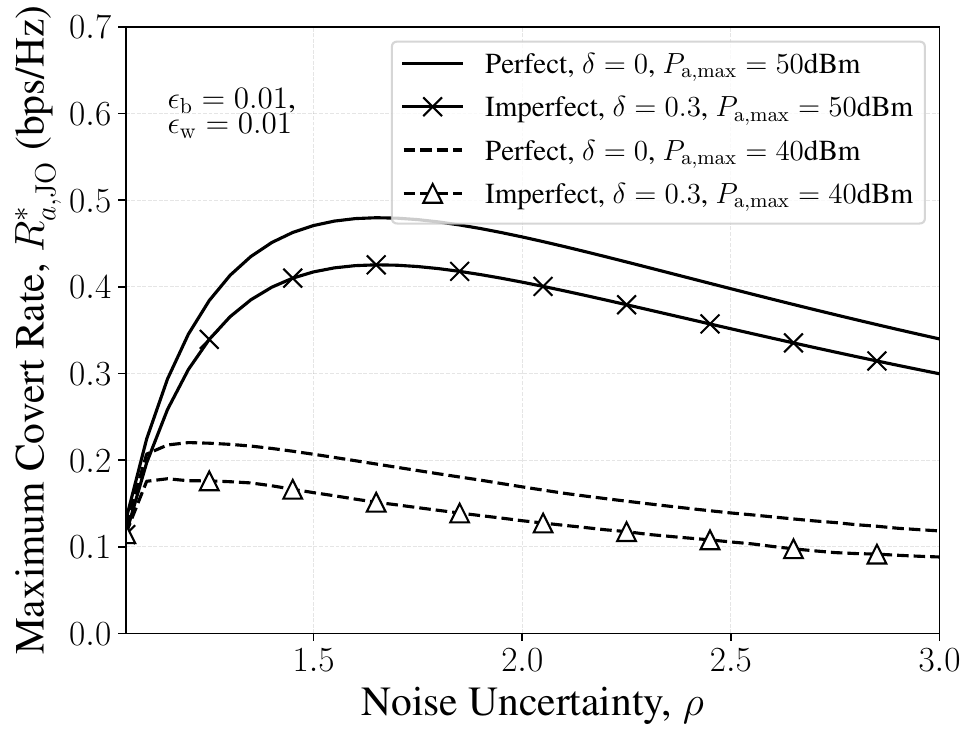}}
        \vspace{-25pt}
        \caption{{Maximum CR achieved by JO-BA vs. $\rho$ under perfect/imperfect channel estimation.}}
        \vspace{-15pt}
        \label{fig:CR_vs_rho_imp}
    \end{minipage}
\end{figure*}

\vspace{-10pt}
\subsection{Impact of Channel Estimation Error} 
\vspace{-2pt}

To explore the effect of the channel estimation error on the maximum CR under OB and JO-BA, we summarize in Fig.~\ref{fig:CR_vs_delta} how the maximum CR varies with $\delta=\delta_{\mathrm b}=\delta_{\mathrm{w}_j}$ for the setting of $P_\mathrm{a,max} = 50$~dBm, $\rho=1.5$, and $\epsilon_{\mathrm{b}}=\epsilon_\mathrm{w} = 0.01$.
It can be seen from Fig.~\ref{fig:CR_vs_delta} that the maximum CR under the two beamforming designs decreases with the increase of $\delta$. 
The reasons behind this phenomenon can be explained as follows.
As $\delta$ increases, the LHS of the constraint \eqref{cons:imperfect_DEP_simp} increases while that of the constraint \eqref{cons:imperfect_TOP_simp} decreases. To satisfy \eqref{cons:imperfect_DEP_simp}, Alice must reduce the covert transmit power $P_{\mathrm{a}}$. Based on this, Alice must reduce the transmit rate $R_{\mathrm{a}}$ to satisfy \eqref{cons:imperfect_TOP_simp}. Thus, the reduced $P_{\mathrm{a}}$ and $R_{\mathrm{a}}$ result in a decrease of the maximum CR.

To explore the effects of the number of antennas at Alice $M_{\mathrm{a}}$ on the maximum CR $R_{\mathrm{a},\mathrm{JO}}^*$ achieved by JO-BA under perfect and imperfect channel estimations, we summarize in Fig.~\ref{fig:CR_vs_antenna_imp} how $R_{\mathrm{a},\mathrm{JO}}^*$ varies with $M_{\mathrm{a}}$ for each setting of $\delta=\{0,\,0.05\}$, $M_{\mathrm{sat}}=\{2\times2,\,8\times8\}$, $P_{\mathrm{a},\max} = 50$ dBm, $\rho=1.5$, and $\epsilon_{\mathrm{b}} = \epsilon_\mathrm{w} = 0.01$.
We can see from Fig.~\ref{fig:CR_vs_antenna_imp} that the maximum CR gap between perfect and imperfect channel estimation becomes more pronounced as $M_{\mathrm{a}}$ or $M_{\mathrm{sat}}$ increases.
The reason behind this phenomenon can be explained as follows.
As $M_{\mathrm{a}}$ increases, the beamforming gain at Bob improves, which significantly increases the negative effects of channel estimation errors according to the LHS of \eqref{cons:imperfect_TOP_simp}, consequently, leading to the maximum CR gap becoming more pronounced.

The impact of the noise uncertainty $\rho$ on the maximum CR $R_{\mathrm{a},\mathrm{JO}}^*$ is further investigated under perfect and imperfect channel estimations.
Under each setting of $\delta=\{0,\,0.05\}$, $P_{\mathrm{a},\max} = \{30,50\}$ dBm, and  $\epsilon_{\mathrm{b}}=\epsilon_{\mathrm{w}}=0.01$, Fig.~\ref{fig:CR_vs_rho_imp} illustrates how $R_{\mathrm{a},\mathrm{JO}}^*$ varies with $\rho$.
We can see from Fig.~\ref{fig:CR_vs_rho_imp} that as $\rho$ increases, the maximum CR gap between perfect and imperfect channel estimation first increases and then remains unchanged.
The reason behind this phenomenon can be explained as follows.
As $\rho$ increases, Alice can adopt a larger transmit power $P_{\mathrm{a}}$, which leads to more significant negative effects of channel estimation according to the LHS of \eqref{cons:imperfect_DEP_simp}, such that the maximum CR gap increases.
When $P_{\mathrm{a}}$ equals $P_{\mathrm{a},\max}$, the negative effects of channel estimation errors remain unchanged, resulting in the maximum CR gap remaining unchanged.

\vspace{-2pt}
\section{Conclusion}
\label{sec:Conclusion}
\vspace{-2pt}
This paper investigated the uplink covert communication in a MIMO GEO satellite-terrestrial system, proposed the optimal beamforming (OB) design as well as the joint optimal beamforming and antenna orientation (JO-BA) design for covert performance enhancement, and conducted the related theoretical modeling for detection error probability (DEP), transmission outage probability (TOP) and covert rate (CR) under both perfect and imperfect channel estimations.
The results in this paper indicate that under both the perfect and imperfect channel estimations, the CR performance can be improved by applying the OB design, and such improvement can be further enhanced by adopting the JO-BA design, especially under the scenario with a moderate number of wardens. In addition, we can have an insight that the CR performance in general improves as a more accurate channel estimation is adopted, and the improvement becomes more significant when a large-scale antenna array is deployed.

\vspace{-2pt}
\appendix
\vspace{-3pt}
\subsection{Proof of Lemma \ref{Theorem:F_z_lb}} \label{Proof:F_z_lb}
\vspace{-2pt}
We first define the following function
\vspace{-1mm}
\begin{equation}
\vspace{-1mm}
\footnotesize
   h_{1,j}(x) \triangleq \frac{\log(\tau_j-x)-\log\left(\hat{\sigma}_{\mathrm{w}_j}^{\mathrm{lb}}\right)}{2\log(\rho)}.
\end{equation}
To analyze the concavity of $h_{1,j}(x)$, we then take the second derivative of $h_{1,j}(x)$ w.r.t. $x$, which is given by
\begin{equation}
\label{eq:partial_g_1}
\footnotesize
\vspace{-1mm}
   \frac{\partial^2 h_{1,j}(x)}{\partial x^2}  = -\frac{1}{2\log(\rho)(x-\tau_j)^2}<0.
\end{equation}
Consequently, $h_{1,j}(x)$ is a concave function and thus satisfies the following inequality. 
\begin{equation}
\vspace{-1mm}
\label{eq:g1_lb}
    \footnotesize
    h_{1,j}(x) \geq \frac{x_2-x}{x_2-x_1} h_{1,j}\left(x_1\right)+\frac{x-x_1}{x_2-x_1} h_{1,j}\left(x_2\right), x \in [x_1,x_2].
\end{equation}
According to \eqref{F_z_w}, we note that $S_{\mathrm{w}_j}\in[0, \tau_j - \hat{\sigma}_{\mathrm{w}_j}^{\mathrm{lb}}]$. 
Let $x_1 = 0$ and $x_2 = \tau_j - \hat{\sigma}_{\mathrm{w}_j}^{\mathrm{lb}}$, the following inequality can be satisfied.

\vspace{-3mm}
\begin{footnotesize}
\begin{align}
    h_{1,j}(S_{\mathrm{w}_j}) 
    \geq \frac{(S_{\mathrm{w}_j}-\tau_j+\hat{\sigma}_{\mathrm{w}_j}^{\mathrm{lb}}) \left(\log (\tau_j)-\log \left(\hat{\sigma}_{\mathrm{w}_j}^{\mathrm{lb}}\right)\right)}{2 \log (\rho) (\hat{\sigma}_{\mathrm{w}_j}^{\mathrm{lb}}- \tau_j)}. \label{eq:F_zw_lower_bound_proof} 
\end{align}    
\end{footnotesize}Based on \eqref{eq:F_zw_lower_bound_proof}, the lower bound of $F_{Z_j}(\tau_j)$ is presented in~\eqref{eq:F_zw_lower_bound}.

We further consider the tightness of the lower bound. Generally, covert communication considers powerful wardens who can conduct an accurate estimation of the noise power, i.e., $\rho$ approaches $1$ \cite{He2017}. Based on this, the interval $(\hat{\sigma}_{s}^{\mathrm{lb}},  \hat{\sigma}_{s}^{\mathrm{ub}})$ is sufficiently narrow, such that $\tau_j \in (\hat{\sigma}_{s}^{\mathrm{lb}}, \hat{\sigma}_{s}^{\mathrm{ub}})$ approaches $\hat{\sigma}_{s}^{\mathrm{lb}}$. Moreover, both $\tau_j - \hat{\sigma}_{\mathrm{w}_j}^{\mathrm{lb}}$ and $S_{\mathrm{w}_j}$ approach 0. Therefore, the equality condition of the \eqref{eq:F_zw_lower_bound_proof} holds in this case, indicating that the lower bound in \eqref{eq:F_zw_lower_bound} is tight.

\vspace{-7pt}
\subsection{Proof of Lemma \ref{Theorem:Opt_tau}} \label{Proof:Opt_tau}
\vspace{-2pt}

To find the optimal detection threshold $\tau_j^*$ that minimizes $\check{\xi}_j$, we take the derivative of $\check{\xi}_j$ w.r.t. $\tau_j$, which is given by
\begin{equation}
\vspace{-1mm}
\footnotesize
    \frac{\partial \check{\xi}_j}{\partial\tau_j} \hspace{-3pt}= \hspace{-3pt}\begin{cases} 0, & \hspace{-2.1mm} \tau_j \hspace{-2pt}<\hspace{-2pt} \hat{\sigma}_{\mathrm{w}_j}^{\mathrm{lb}}, \\  
        \frac{\omega_j \left( \mu_j(\tau_j) \Gamma \left(m_j,\mu_j(\tau_j)\right)-  \gamma \left(m_j+1,\mu_j(\tau_j)\right) \nu_j(\tau_j) )\right)}{2  \tau_j\log(\rho)\Gamma(m_j+1)(\tau_j-\hat{\sigma}_{\mathrm{w}_j}^{\mathrm{lb}})}, & \hspace{-2.1mm}\hat{\sigma}_{\mathrm{w}_j}^{\mathrm{lb}} \hspace{-2pt}\le\hspace{-2pt} \tau_j \le \hat{\sigma}_{\mathrm{w}_j}^{\mathrm{ub}},\\  \frac{\omega_j \left( \mu_j(\tau_j) \gamma \left(m_j,\mu_j(\tau_j)\right)+  \gamma \left(m_j+1,\mu_j(\tau_j)\right) \nu_j(\tau_j) )\right)}{2\tau_j\log(\rho)\Gamma(m_j+1)(\tau_j-\hat{\sigma}_{\mathrm{w}_j}^{\mathrm{lb}})}, & \hspace{-2.1mm}\hat{\sigma}_{\mathrm{w}_j}^{\mathrm{ub}}\hspace{-2pt}<\hspace{-2pt}\tau_j. \end{cases}
\end{equation}

Note that $\mu_j(\tau_j)>0$ and $\nu_j(\tau_j)>0$, thus $\frac{\partial \check{\xi}_j}{\partial\tau_j} >0$ when $\tau_j>\hat{\sigma}_{\mathrm{w}_j}^{\mathrm{ub}}$, indicating that $\check{\xi}_j$ monotonically increases with $\tau_j$.

Next, we explore the monotonicity of $\check{\xi}_j$ within the interval $\hat{\sigma}_{\mathrm{w}_j}^{\mathrm{lb}} \le \tau_j \le \hat{\sigma}_{\mathrm{w}_j}^{\mathrm{ub}}$. Define $h_{2,j}(\tau_j) \triangleq \mu_j(\tau_j) \Gamma \left(m_j,\mu_j(\tau_j)\right)- \gamma \left(m_j+1,\mu_j(\tau_j)\right) \nu_j(\tau_j)$,
we have $\frac{\partial \check{\xi}_j}{\partial\tau_j}=\frac{\omega_j h_{2,j}(\tau_j)}{2\tau_j\log(\rho)\Gamma(m_j+1)(\tau_j-\hat{\sigma}_{\mathrm{w}_j}^{\mathrm{lb}})}$. Setting $\frac{\partial \check{\xi}_j}{\partial\tau_j}=0$, i.e., $h_{2,j}(\tau_j)=0$, we can have
\begin{equation}
\label{eq:tau_zero_point}
\footnotesize
    \frac{\mu_j(\tau_j) \Gamma \left(m_j,\mu_j(\tau_j)\right)}{\gamma \left(m_j+1,\mu_j(\tau_j)\right)}= \nu_j(\tau_j).
\end{equation}

To determine the solutions of \eqref{eq:tau_zero_point}, we first analyze the monotonicity of the LHS and right-hand side (RHS) of \eqref{eq:tau_zero_point}, separately. Let $h_{3,j}(\tau_j)\triangleq\frac{\mu_j(\tau_j) \Gamma \left(m_j,\mu_j(\tau_j)\right)}{\gamma \left(m_j+1,\mu_j(\tau_j)\right)}$, the first derivative of $h_{3,j}(\tau_j)$ w.r.t. $\tau_j$ is given by
\begin{equation}
\footnotesize
    \frac{\partial h_{3,j}(\tau_j)}{\partial \tau_j}=\frac{\partial h_{3,j}(\tau_j)}{\partial \mu_j(\tau_j)}\times\frac{\partial \mu_j(\tau_j)}{\partial \tau_j},
\end{equation}
where 

\vspace{-3.5mm}
\begin{footnotesize}
\begin{align}
    \frac{\partial h_{3,j}(\tau_j)}{\partial\mu_j(\tau_j)}=&\frac{h_{4,j}(\tau_j)\Gamma(m_j,\mu_j(\tau_j)) }{\gamma \left(m_j+1,\mu_j(\tau_j)\right)^2},\\
    h_{4,j}(\tau_j) \triangleq&-\left(\mu_j(\tau_j)\Gamma(m_j,\mu_j(\tau_j))+\gamma(m_j+1,\mu_j(\tau_j))\right)\nonumber\\
    & \times \frac{\exp(-\mu_j(\tau_j))\mu_j(\tau_j)^{m_j}}{\Gamma(m_j,\mu_j(\tau_j))}+\gamma(m_j+1,\mu_j(\tau_j)),    \\
    \frac{\partial \mu_j(\tau_j)}{\partial \tau_j} =& \frac{m_j}{\omega_j}>0. \label{eq:mu_monotonic}
\end{align}    
\end{footnotesize}

Note that the sign of $\frac{\partial h_{3,j}(\tau_j)}{\partial\mu_j(\tau_j)}$ depends on that of $h_{4,j}(\tau_j)$, which remains difficult to determine directly. Thus, we further take the first derivative of $h_{4,j}(\tau_j)$ w.r.t. $\mu_j(\tau_j)$, which can be given by
\begin{equation}
    \footnotesize
    \hspace{-2.5mm}\frac{\partial h_{4,j}(\tau_j)}{\partial\mu_j(\tau_j)}=\frac{h_{5,j}(\tau_j)(\mu_j(\tau_j)\Gamma(m_j,\mu_j(\tau_j))+\gamma(m_j+1,\mu_j(\tau_j))) }{\exp(2\mu_j(\tau_j))\mu_j(\tau_j)^{1-2m_j}\Gamma \left(m_j,\mu_j(\tau_j)\right)^2},\hspace{-1mm}\label{eq:partial_h3}
\end{equation}
where
\begin{equation}
\footnotesize
    h_{5,j}(\tau_j) \triangleq\mu_j(\tau_j)^{-m_j}\exp(\mu_j(\tau_j))(\mu_j(\tau_j)-m_j)\Gamma(m_j,\mu_j(\tau_j))-1. \label{eq:partial_h5}
\end{equation}
Following \eqref{eq:partial_h3}, the sign of ${\partial h_{4,j}(\tau_j)}/{\partial\mu_j(\tau_j)}$ depends on that of $h_{5,j}(\tau_j)$. When $\mu_j(\tau_j)\leq m_j$, $h_{5,j}(\tau_j)<0$. When $\mu_j(\tau_j)>m_j$, according to the inequality in [\citenum{Borwein2009}, Eq. (2.1)], we have

\vspace{-3.5mm}
\begin{footnotesize}
\begin{align}
    h_{5,j}(\tau_j) = & \left(\mu\left(\tau_j\right)-m_j\right)\int_0^\infty\exp\left(-\mu\left(\tau_j\right)t\right)\left(1+t\right)^{m_j-1}dt  -1   \nonumber \\  \leq& \frac{\mu_j(\tau_j)-m_j}{\mu_j(\tau_j)-m_j+1}-1 <0. 
\end{align}    
\end{footnotesize}Consequently, $h_{5,j}(\tau_j)<0$ always holds leading to ${\partial h_{4,j}(\tau_j)}/{\partial\mu_j(\tau_j)}<0$, which indicates that $h_{4,j}(\tau_j)$ monotonically decreases with $\mu_j(\tau_j)$. 
Note that $\lim_{\mu_j(\tau_j)\to0} h_{4,j}(\tau_j)=0$, thus $h_{4,j}(\tau_j)<0$ for all $\mu_j(\tau_j)>0$, which results in ${\partial h_{3,j}(\tau_j)}/{\partial\mu_j(\tau_j)}<0$. 
Therefore, $h_{3,j}(\tau_j)$ is a monotonically decreasing function of $\tau_j$, with $\lim_{\tau_j\to\hat{\sigma}_{\mathrm{w}_j}^{\mathrm{lb}}} h_{3,j}(\tau_j)=+\infty$ and $\lim_{\tau_j\to+\infty} h_{3,j}(\tau_j)=0$.

Next, we analyze the monotonicity of the RHS of \eqref{eq:tau_zero_point}. The first derivative of $\nu_j(\tau_j)$ w.r.t. $\tau_j$ is given by
\begin{equation}
\footnotesize
    \frac{\partial \nu_j(\tau_j)}{\partial\tau_j} = \frac{\tau_j/\hat{\sigma}_{\mathrm{w}_j}^{\mathrm{lb}}-1-\log(\tau_j/\hat{\sigma}_{\mathrm{w}_j}^{\mathrm{lb}})}{(\tau_j/\hat{\sigma}_{\mathrm{w}_j}^{\mathrm{lb}}-1)^2}.
\end{equation}
According to the logarithmic inequality $x-1\geq\log(x), \forall x\geq0$, we have ${\partial \nu_j(\tau_j)}/{\partial\tau_j}>0$. Therefore, $\nu_j(\tau_j)$ is a monotonically increasing function of $\tau_j$. In addition, we have $\lim_{\tau_j\to\hat{\sigma}_{\mathrm{w}_j}^{\mathrm{lb}}} \nu_j(\tau_j)=0$ and $\lim_{\tau_j\to+\infty} \nu_j(\tau_j)=+\infty$.

According to the above analysis about LHS and RHS of \eqref{eq:tau_zero_point}, we know that there is a unique solution for \eqref{eq:tau_zero_point} over $[\hat{\sigma}_{\mathrm{w}_j}^{\mathrm{lb}},+\infty)$, which is denoted by $\tau'_j$.
If $\tau_j' > \hat{\sigma}_{\mathrm{w}_j}^{\mathrm{ub}}$, $\check{\xi}_j$ is monotonically decreasing over $[\hat{\sigma}_{\mathrm{w}_j}^{\mathrm{lb}} ,\hat{\sigma}_{\mathrm{w}_j}^{\mathrm{ub}}]$, thus the optimal detection threshold $\tau_j^*=\hat{\sigma}_{\mathrm{w}_j}^{\mathrm{ub}}$.  
If $\hat{\sigma}_{\mathrm{w}_j}^{\mathrm{lb}} \le \tau_j' \le \hat{\sigma}_{\mathrm{w}_j}^{\mathrm{ub}}$, $\check{\xi}_j$ is monotonically decreasing over $[\hat{\sigma}_{\mathrm{w}_j}^{\mathrm{lb}}, \tau_j']$ and monotonically increasing over $(\tau_j', \hat{\sigma}_{\mathrm{w}_j}^{\mathrm{ub}}]$, therefore $\tau_j^*=\tau_j'$.

\vspace{-7pt}
\subsection{Proof of Theorem \ref{Theorem:DEP_lower_bound}} \label{Proof:DEP_lower_bound}
\vspace{-2pt}
Let $h_{6,j}(x) \triangleq x \Gamma(m_j,x)-\Gamma(m_j+1,x)$, \eqref{eq:xi_star} can be rewritten as

\vspace{-4mm}
\begin{footnotesize}
\begin{align}
\label{eq:DEP_lower_proof_1}
    \check{\xi}_j^* =
     & 1-\frac{ \omega_j\log \left({ \tau_j^*}/{\hat{\sigma}_{\mathrm{w}_j}^{\mathrm{lb}}}\right) \left(h_{6,j}(\mu_j(\tau_j^*)) + \Gamma \left(m_j+1\right)\right)}{2 \log (\rho) \Gamma (m_j+1) (\tau_j^*-\hat{\sigma}_{\mathrm{w}_j}^{\mathrm{lb}})}. 
\end{align}    
\end{footnotesize}Taking the first derivative of $h_{6,j}(x)$ w.r.t. $x$, we can have
\begin{equation}
\footnotesize
    \frac{\partial h_{6,j}(x)}{\partial x} = \Gamma(m_j,x) >0,
\end{equation}
indicating that $h_{6,j}(x)$ monotonically increases with $x$. Additionally, we have ${ h_{6,j}(0)}=-\Gamma(m_j+1)$ and $\lim_{_{x\to \infty}}{h_{6,j}(x)}=0$. Then, we can conclude that $h_{6,j}(x)<0$ for all $x\geq0$. 
By applying $h_{6,j}(x)=0$ to \eqref{eq:DEP_lower_proof_1}, we can have 

\vspace{-3mm}
\begin{footnotesize}
\begin{align}
\label{eq:DEP_lower_proof_2}
    \check{\xi}_j^*
    > & 1-\frac{\omega_j \left(\log \left({ \tau_j^*}\right)-\log(\hat{\sigma}_{\mathrm{w}_j}^{\mathrm{lb}})\right)}{2 \log (\rho) ( \tau_j^*-\hat{\sigma}_{\mathrm{w}_j}^{\mathrm{lb}})}. 
\end{align}    
\end{footnotesize}

Let $f(x)=\log x$. According to the mean value theorem, there exists at least one point $c\in(\hat{\sigma}_{\mathrm{w}_j}^{\mathrm{lb}},\tau_j^*)$ such that $f'(c)=\frac{1}{c}=\frac{f(\tau_j^*)-f(\hat{\sigma}_{\mathrm{w}_j}^{\mathrm{lb}})}{\tau_j^*-\hat{\sigma}_{\mathrm{w}_j}^{\mathrm{lb}}}$.
Since $\frac{1}{c}<\frac{1}{\hat{\sigma}_{\mathrm{w}_j}^{\mathrm{lb}}}$, it follows that  $\check{\xi}_j^*>1-\frac{\omega_j}{2\log(\rho)\hat{\sigma}_{\mathrm{w}_j}^{\mathrm{lb}}}$, which completes the proof.

Following this, we continue to discuss the tightness of the lower bound $\check{\xi}^{\mathrm{lb}}_{j}$.
Firstly, the tightness of \eqref{eq:DEP_lower_proof_2} is considered. To achieve covert communication, the second term of the RHS of \eqref{eq:DEP_lower_proof_1} must approach $0$, which causes DEP to approach $1$. Accordingly, $\omega_j$ in \eqref{eq:DEP_lower_proof_1} is required to approach $0$.
Based on this, as $\omega_j \to 0$,  we have $\lim_{_{\omega_j\to 0}}\mu_j(\tau_j^*)=\infty$ and $\lim_{_{\omega_j\to 0}}{h_{6,j}(\mu_j(\tau_j^*))}=0$. Thus, the lower bound in \eqref{eq:DEP_lower_proof_2} is tight.
Next, we explore the tightness of \eqref{eq:DEP_lower_bound}. 
Generally, covert communication assumes powerful wardens who have a strong capability to estimate noise power, i.e., $\rho$ is small \cite{He2017}. 
Consequently, the interval $(\hat{\sigma}_{\mathrm{w}_j}^{\mathrm{lb}},\tau_j^*)$ is sufficiently narrow, such that $1/c \to 1/\hat{\sigma}_{\mathrm{w}_j}^{\mathrm{lb}}$, indicating that the lower bound in \eqref{eq:DEP_lower_bound} is also tight.

\vspace{-7pt}
\subsection{Proof of Theorem \ref{Theorem:minimum_psi_ab}} \label{Proof:minimum_psi_ab}
\vspace{-2pt}
Following problem~\eqref{Prob:antenna_perfect_Gab}, we can obtain the optimal power can be expressed as $P_\mathrm{a}^* =\min( P_{\mathrm{a},\max}, P_{\mathrm{a}}^{\mathrm{ub}})$, where $P_{\mathrm{a}}^{\mathrm{ub}}=\min_j$\scalebox{0.85}{$\big({\eta_{\mathrm{w}_j}}/{(G_{\mathrm{a},\mathrm{w}_j}{(1+K\,\mathrm{Tr}(\mathbf{V}_{\mathrm{w}_j}\overline{\mathbf{H}}_{\mathrm{a},\mathrm{w}_j}\widehat{\mathbf{W}}_{\mathrm{a}}\overline{\mathbf{H}}_{\mathrm{a},\mathrm{w}_j}^{\dagger}))})} \big)$}.
To obtain the optimal antenna boresight vector, we consider the following two cases: $P_\mathrm{a}^* = P_{\mathrm{a},\max} \neq P_{\mathrm{a}}^{\mathrm{ub}}$ and $P_\mathrm{a}^* = P_{\mathrm{a}}^{\mathrm{ub}}$.

When $P_\mathrm{a}^* = P_{\mathrm{a},\max} \neq P_{\mathrm{a}}^{\mathrm{ub}}$, it implies that $P_{\mathrm{a},\max} < P_{\mathrm{a}}^{\mathrm{ub}}$. In this case, the constraint~\eqref{cons:covert_req_antenna} is always satisfied. Accordingly, Alice can adopt the maximum transmit power $P_{\mathrm{a},\max}$ and maximum antenna gain $G_{\mathrm{a},\max}^{\text{dBi}}$ to transmit covert signals, i.e., $G_{\mathrm{a},\mathrm{b}}^{\text{dBi}} = G_{\mathrm{a},\max}^{\text{dBi}}$. According to \eqref{antenna_gain_alice}, we can have that the optimal antenna boresight vector $\mathbf{o}_{\mathrm{a}}^*$ should satisfy $\vartheta_{\mathrm{a},\mathrm{b}} \le \vartheta_0$.

When $P_\mathrm{a}^* = P_{\mathrm{a}}^{\mathrm{ub}}$, it implies that $P_{\mathrm{a},\max} \geq P_{\mathrm{a}}^{\mathrm{ub}}$. Substituting $P_\mathrm{a}^*$ into the objective function \eqref{obj:Gab_Pa_wa}, problem~\eqref{Prob:antenna_perfect_Gab} can be simplified as 
\vspace{-1pt}
\begin{equation}
\footnotesize
    \label{Prob:antenna_perfect_reduce}
     \max_{ \mathbf{o}_{\mathrm{a}}}\;   G_{\mathrm{a},\mathrm{b}}^{\text{dBi}}-G_{\mathrm{a},\mathrm{w}_{j^*}}^{\text{dBi}}.
\end{equation}
where $j^* = \arg\min_j$\scalebox{0.85}{$\big({\eta_{\mathrm{w}_j}}/{G_{\mathrm{a},\mathrm{w}_j}{(1+K\,\mathrm{Tr}(\mathbf{V}_{\mathrm{w}_j}\overline{\mathbf{H}}_{\mathrm{a},\mathrm{w}_j}\widehat{\mathbf{W}}_{\mathrm{a}}\overline{\mathbf{H}}_{\mathrm{a},\mathrm{w}_j}^{\dagger}))}} \big)$}.
To solve problem \eqref{Prob:antenna_perfect_reduce}, the optimal antenna orientation should maximize the difference between $G_{\mathrm{a},\mathrm{b}}^{\text{dBi}}$ and $G_{\mathrm{a},\mathrm{w}_{j^*}}^{\text{dBi}}$. In this case, the default antenna orientation $\mathbf{o}_{\mathrm{a}} = \mathbf{q}_{\mathrm{b}} - \mathbf{q}_{\mathrm{a}}$ is a feasible solution to problem~\eqref{Prob:antenna_perfect_reduce}, and it satisfies $G_{\mathrm{a},\mathrm{b}}^{\text{dBi}} = G_{\mathrm{a},\max}^{\text{dBi}}$, where $G_{\mathrm{a},\max}^{\text{dBi}} > G_{\mathrm{a},\mathrm{w}_{j^*}}^{\text{dBi}}$.
Thus, the optimal antenna orientation $\mathbf{o}_{\mathrm{a}}^*$ should satisfy $G_{\mathrm{a},\mathrm{b}}^{\text{dBi}} > G_{\mathrm{a},\mathrm{w}_{j^*}}^{\text{dBi}} $, i.e., $\vartheta_{\mathrm{a},\mathrm{b}}<\vartheta_{\mathrm{a},\mathrm{w}_{j^*}}$. 

To further prove the optimal antenna orientation $\mathbf{o}_{\mathrm{a}}^*$ satisfies Theorem \ref{Theorem:minimum_psi_ab} under the condition $P_\mathrm{a}^* = P_{\mathrm{a}}^{\mathrm{ub}}$, we proceed by contradiction. In particular, we assume the antenna boresight vector $\mathbf{o}_{\mathrm{a}}^*$ satisfying $\vartheta_{\mathrm{a},\mathrm{b}} \neq \vartheta_0$, which can be discussed in the following two cases: $\vartheta_{\mathrm{a},\mathrm{b}} > \vartheta_0$ and $\vartheta_{\mathrm{a},\mathrm{b}} < \vartheta_0$.

Under the assumption that $\mathbf{o}_{\mathrm{a}}^*$ satisfies $\vartheta_{\mathrm{a},\mathrm{b}} > \vartheta_0$, we can adjust $\mathbf{o}_{\mathrm{a}}^*$ leading to $\vartheta_{\mathrm{a},\mathrm{b}}$ decrease by an arbitrarily small value $\Delta \vartheta_{\mathrm{a},\mathrm{b}}$, i.e., $\Delta \vartheta_{\mathrm{a},\mathrm{b}}<0$. 
As $\mathbf{o}_{\mathrm{a}}^*$ is adjusted, the off-boresight angle of the warden $j^*$ also varies. The corresponding variation is denoted by $\Delta \vartheta_{\mathrm{a},\mathrm{w}_{j^*}}$ and can be given by $\Delta \vartheta_{\mathrm{a},\mathrm{w}_{j^*}}=\cos\varTheta_{j^*} \Delta \vartheta_{\mathrm{a},\mathrm{b}}$, where $\cos\varTheta_{j^*} =\frac{(\mathbf{o}_{\mathrm{a}}^* \times \mathbf{u}_{\mathrm{a},\mathrm{b}})\cdot(\mathbf{o}_{\mathrm{a}}^* \times \mathbf{u}_{\mathrm{a},\mathrm{w}_{j^*}})}{\left \|\mathbf{o}_{\mathrm{a}}^* \times \mathbf{u}_{\mathrm{a},\mathrm{b}}\right \|\cdot\|\mathbf{o}_{\mathrm{a}}^* \times \mathbf{u}_{\mathrm{a},\mathrm{w}_{j^*}} \|} $~[\citenum{Chauvenet1876}, Eq. (285)].
According to \eqref{antenna_gain_alice}, the variation in $G_{\mathrm{a},\mathrm{b}}^{\text{dBi}}$ and $G_{\mathrm{a},\mathrm{w}_{j^*}}^{\text{dBi}}$ can be expressed as
\begin{equation}
    \footnotesize
    \hspace{-9mm}\Delta G_{\mathrm{a},\mathrm{b}}^{\text{dBi}}  = \frac{\partial G_{\mathrm{a},\mathrm{b}}^{\text{dBi}}}{\partial \vartheta_{\mathrm{a},\mathrm{b}}} \Delta \vartheta_{\mathrm{a},\mathrm{b}}=\begin{cases} \frac{-25\Delta \vartheta_{\mathrm{a},\mathrm{b}}}{\vartheta_{\mathrm{a},\mathrm{b}}\log(10)},&\hspace{-2mm} \vartheta_0\leq\vartheta_{\mathrm{a},\mathrm{b}}\leq 48^{\circ},\\
        0, & \hspace{-2mm}\text{otherwise},\end{cases} \label{eq:partial_Gab}
\end{equation}
\begin{equation}
    \footnotesize
    \hspace{0mm}\Delta G_{\mathrm{a},\mathrm{w}_{j^*}}^{\text{dBi}} \hspace{-2pt}= \frac{\partial G_{\mathrm{a},\mathrm{w}_{j^*}}^{\text{dBi}}}{\partial \vartheta_{\mathrm{a},\mathrm{w}_{j^*}}} \Delta \vartheta_{\mathrm{a},\mathrm{w}_{j^*}} \hspace{-2pt}= \begin{cases}
        \frac{-25\Delta \vartheta_{\mathrm{a},\mathrm{w}_{j^*}}}{\vartheta_{\mathrm{a},\mathrm{w}_{j^*}}\log(10)}, & \hspace{-2.5mm}\vartheta_0\leq\vartheta_{\mathrm{a},\mathrm{w}_{j^*}}\leq 48^{\circ}, \label{eq:partial_Gaw} \\
        0, &\hspace{-2.5mm} \text{otherwise}.
    \end{cases}
\end{equation}
Consider the worst case that the adjustment of $\mathbf{o}_{\mathrm{a}}^*$ leads to the maximum $\Delta G_{\mathrm{a},\mathrm{w}_{j^*}}^{\text{dBi}}$, i.e., $\cos\varTheta_{j^*}=1$. Thus, $\Delta \vartheta_{\mathrm{a},\mathrm{w}_{j^*}}=\Delta \vartheta_{\mathrm{a},\mathrm{b}}$. 
Since $\vartheta_{\mathrm{a},\mathrm{b}}<\vartheta_{\mathrm{a},\mathrm{w}_{j^*}}$, it follows that $\Delta G_{\mathrm{a},\mathrm{b}}^{\text{dBi}} > \Delta G_{\mathrm{a},\mathrm{w}_{j^*}}^{\text{dBi}}\geq 0$.
In this case, an adjustment to $\mathbf{o}_{\mathrm{a}}^*$ that reduces $\vartheta_{\mathrm{a},\mathrm{b}}$ until $\vartheta_{\mathrm{a},\mathrm{b}} \leq \vartheta_0$ leads to $\Delta G_{\mathrm{a},\mathrm{b}}^{\text{dBi}} - \Delta G_{\mathrm{a},\mathrm{w}{j^*}}^{\text{dBi}} > 0$. This implies that $\mathbf{o}_{\mathrm{a}}^*$ can be further adjusted to increase the objective function until $\vartheta_{\mathrm{a},\mathrm{b}} \leq \vartheta_0$. Thereby, the results contradict the initial assumption that the optimal antenna boresight vector $\mathbf{o}_{\mathrm{a}}^*$ satisfies $\vartheta_{\mathrm{a},\mathrm{b}} > \vartheta_0$.

Under the assumption that the optimal $\mathbf{o}_{\mathrm{a}}^*$ satisfies $\vartheta_{\mathrm{a},\mathrm{b}}<\vartheta_0$, 
we can adjust $\mathbf{o}_a^*$ leading to $\vartheta_{\mathrm{a},\mathrm{w}_{j^*}}$ increasing with an arbitrarily small positive value $\Delta \vartheta_{\mathrm{a},\mathrm{w}_{j^*}}>0$, as well as leading to $ \vartheta_{\mathrm{a},\mathrm{b}}$ varying with a value $\Delta \vartheta_{\mathrm{a},\mathrm{b}}$ and $\Delta \vartheta_{\mathrm{a},\mathrm{b}}=\cos\varTheta_{j^*}\Delta\vartheta_{\mathrm{a},\mathrm{w}_{j^*}}$.
In this case, we can have $\Delta G_{\mathrm{a},\mathrm{b}}^{\text{dBi}}=0$ and $\Delta G_{\mathrm{a},\mathrm{w}_{j^*}}^{\text{dBi}}\leq 0$. Consequently, $\Delta G_{\mathrm{a},\mathrm{w}_{j^*}}^{\text{dBi}}\leq \Delta G_{\mathrm{a},\mathrm{b}}^{\text{dBi}}$, and thus $\Delta G_{\mathrm{a},\mathrm{b}}^{\text{dBi}} - \Delta G_{\mathrm{a},\mathrm{w}{j^*}}^{\text{dBi}} \geq 0$, which implies that $\mathbf{o}_a^*$ can be further adjusted to improve the objective function until $\vartheta_{\mathrm{a},\mathrm{b}}\geq\vartheta_0$.
Therefore, the above results contradict the initial assumption that the antenna boresight vector $\mathbf{o}_{\mathrm{a}}^*$ satisfies $\vartheta_{\mathrm{a},\mathrm{b}} < \vartheta_0$.

Based on the above results, the optimal antenna boresight vector $\mathbf{o}_{\mathrm{a}}^*$ should satisfy $\vartheta_{\mathrm{a},\mathrm{b}} = \vartheta_0$ under both $P_\mathrm{a}^* = P_{\mathrm{a},\max} \neq P_{\mathrm{a}}^{\mathrm{ub}}$ and $P_\mathrm{a}^* = P_{\mathrm{a}}^{\mathrm{ub}}$.

\vspace{-9pt}

\bibliographystyle{IEEEtran}
\bibliography{satellite_covert_uplink.bib}

\end{document}